\documentclass[journal]{IEEEtran}
\usepackage{float}% 将表格用H参数控制在指定位置
\usepackage{amssymb}
\usepackage{bm}
\usepackage{amsthm, amsmath,amsfonts}
\usepackage[ruled, vlined]{algorithm2e}
\usepackage{array}
\usepackage{textcomp}
\usepackage{stfloats}
\usepackage{url}
\usepackage{verbatim}
\usepackage{graphicx}
\usepackage{subfigure}
\usepackage{cite}
\usepackage{booktabs}    %导言区
\usepackage{mathrsfs} % 花体	
\usepackage{setspace}% 行间距调整
\usepackage{geometry}% 调整页边距
\usepackage{multirow}% 表格使用多行
\usepackage{amsmath} % 公式换行
\allowdisplaybreaks[4] % 公式换行
\usepackage{autobreak} % 公式换行
\usepackage{xcolor}
\usepackage{hyperref} % 引用发光
\usepackage{geometry}  % 调整页眉
\usepackage{color}
\usepackage{threeparttable}

\newtheorem{Remark}{Remark}% number Remark
\theoremstyle{definition}
\newtheorem{Theorem}{Theorem}% number Theorem
\newtheorem{Corollary}{Corollary}% number Corollary
\newtheorem{Lemma}{Lemma}% number Lemma
\makeatletter
\renewcommand{\maketag@@@}[1]{\hbox{\m@th\normalsize\normalfont#1}}%
\makeatother

\makeatletter
\let\myorg@bibitem\bibitem
\def\bibitem#1#2\par{%
	\@ifundefined{bibitem@#1}{%
		\myorg@bibitem{#1}#2\par
	}{%
		\begingroup
		\color{\csname bibitem@#1\endcsname}%
		\myorg@bibitem{#1}#2\par
		\endgroup
	}%
}
\makeatother
\begin{document}
\begin{spacing}{1}
%\bibliographystyle{alpha}% Number the references in the order of their appearance.
% \linenumbers% 插入行号
% \pagewiselinenumbers% 按页重头
% 分文件编写

\title{
Unauthorized UAV Countermeasure for Low-Altitude Economy: Joint Communications and Jamming based on MIMO Cellular Systems
}
\author{Zhuoran Li, Zhen Gao, Kuiyu Wang, Yikun Mei, Chunli Zhu, Lei Chen, Xiaomei Wu, Dusit Niyato \textit{Fellow}, \textit{IEEE}
	% <-this % stops a space
	%		\thanks{}% <-this % stops a space
	%		This paper was produced by the IEEE Publication Technology Group. They are in Piscataway, NJ.
	\thanks{	
	Zhuoran Li, Kuiyu Wang, Yikun Mei, and Lei Chen are with the School of Information and Electronics, Beijing Institute of Technology, Beijing 100081, China, and also with the Advanced Research Institue of Multidisciplinary Sciences, Beijing Institute of Technology, Beijing 100081, China
		(e-mails: \mbox{lizhuoran2000@qq.com},
		\mbox{3220215125@bit.edu.cn},
		\mbox{meiyikun@bit.edu.cn},		
		\mbox{leichen@bit.edu.cn}
		).
	
	Zhen Gao is with State Key Laboratory of CNS/ATM, Beijing Institute of Technology (BIT), Beijing 100081, China, also with BIT Zhuhai 519088, China, also with the MIT Key Laboratory of Complex-Field Intelligent Sensing, BIT, Beijing 100081, China, also with the Advanced Technology Research Institute of BIT (Jinan), Jinan 250307, China, and also with the Yangtze Delta Region Academy, BIT (Jiaxing), Jiaxing 314019, China
	(e-mails: \mbox{gaozhen16@bit.edu.cn}).\textit{(Corresponding author: Zhen Gao.)}
	
	Chunli Zhu is with the State Key Laboratory of CNS/ATM \& State Key
	Laboratory of Explosion Science and Safety Protection, Beijing Institute of
	Technology, Beijing 100081, China (e-mail: chunlizhu@bit.edu.cn).
	
	Xiaomei Wu is with China Tower Corporation Limited, No.9 Dongran North Street, Haidian District, Beijing 100089, China (e-mail: wuxm@chinatowercom.cn).
	
	Dusit Niyato is with the College of Computing and Data Science, Nanyang
	Technological University, Singapore 639798, (e-mail: dniyato@ntu.edu.sg).
	}
}
%	\vspace{-10mm}
%	\markboth{}%
%	Journal of \LaTeX\ Class Files,~Vol.~14, No.~8, August~2021
%	{Shell \MakeLowercase{\textit{et al.}}: A Sample Article Using IEEEtran.cls for IEEE Journals}

%	\IEEEpubid{}
%	0000--0000/00\$00.00~\copyright~2021 IEEE

\maketitle
\begin{abstract}	
To ensure the thriving development of low-altitude economy, countering unauthorized unmanned aerial vehicles (UAVs) is an essential task.
The existing widely deployed base stations hold great potential for joint communication and jamming.
In light of this, this paper investigates the joint design of beamforming to simultaneously support communication with legitimate users and countermeasure against unauthorized UAVs based on dual-functional multiple-input multiple-output (MIMO) cellular systems.	
We first formulate a joint communication and jamming (JCJ) problem, relaxing it through semi-definite relaxation (SDR) to obtain a tractable semi-definite programming (SDP) problem, with SDR providing an essential step toward simplifying the complex JCJ design.
Although the solution to the relaxed SDP problem cannot directly solve the original problem, it offers valuable insights for further refinement. 
Therefore, we design a novel constraint specifically tailored to the structure of the SDP problem, ensuring that the solution adheres to the rank-1 constraint of the original problem.
Finally, we validate effectiveness of the proposed JCJ scheme through extensive simulations.
Simulation codes are provided to reproduce the results in this paper: \href{https://github.com/LiZhuoRan0}{https://github.com/LiZhuoRan0}.
The results confirm that the proposed JCJ scheme can operate effectively when the total number of legitimate users and unauthorized UAVs exceeds the number of antennas.
\end{abstract}

\begin{IEEEkeywords}
	low-altitude economy, countermeasure, unmanned aerial vehicle, multiple-input multiple-output, joint communication and jamming
\end{IEEEkeywords}

\section{Introduction}\label{sec_Introduction}
\IEEEPARstart{T}{he}  concept of low-altitude economy (LAE) is conceived as an economic system that integrates a variety of low-altitude aviation activities involving both unmanned and manned craft, including unmanned aerial vehicles (UAVs) and electric vertical take-off and landing (eVTOL) aircraft\cite{ref_arXiv_IllegalUAV_NoJamming}.
Anticipated to facilitate a multitude of low-altitude services across sectors such as transportation, environmental surveillance, agriculture, and entertainment, LAE holds the promise of significant economic and social benefits.
This potential has sparked a surge in research interest\cite{ref_UAV_YongZeng,ref_CJE_UAV_Networks,wang2024maximizing,ref_CJE_IoT_DRL,ref_SpaceSciTechnol_0176}.

For LAE to be successful, it is essential not only to ensure continuous wireless communication with authorized UAVs, but also to implement necessary countermeasures against unauthorized UAVs\cite{ref_AESM_CounterUAV, ref_Access_CounterUAV}.
Generally, there are two kinds of countermeasure techniques including  physical capture and jamming\cite{ref_Conf_UAV_CommCounterPhysicalCounter}.
Physical capture approaches are efficient and low-cost, but not friendly to pilots.
Jamming is the most popular method used in neutralizing UAVs entering restricted areas, and there are four categories of jamming\cite{ref_AESM_CounterUAV}.
The first category is emitting noise-like unstructured signals to reduce the signal-to-interference-plus-noise ratio (SINR) of the unauthorized UAVs.
Partial-band noise jamming was implemented to cut off the control link between the UAV and its controller in \cite{ref_Conf_PartialBandNoise}.
The second category entails broadcasting deceptive Global Position System (GPS) signals to lead unauthorized UAVs out of restricted areas.
The third category involves hacking the communication protocols of the unauthorized UAVs and generating fake commands to control them to fly away restricted areas\cite{ref_Conf_IdentifyUAVSigToJam, ref_IdentifyUAVSig_1, ref_IdentifyUAVSig_2, ref_IdentifyUAVSig_3}.
The fourth category is passive interference, which involves using passive devices to modify the signals in the environment, causing multipath signals to superpose destructively at the receiver.
A common approach is to design the phase of a reconfigurable intelligent surface (RIS)\cite{refWclRisJamming,refIoTJDestructiveRIS,refTVTillegalRISchannelAging,ref_IETComm_UAV_Secure}.

Traditional counter-UAV methods, which require dedicated devices, platforms, and architectures \cite{ref_Access_CounterUAV}, significantly increase the cost of countering unauthorized UAVs. However, the potential of widely deployed multiple-input multiple-output (MIMO) base stations (BSs) has not been fully utilized for this purpose. With the evolution of dense cell infrastructures into networks with integrated sensing and communication (ISAC) capabilities, referred to as perceptive mobile networks (PMNs) \cite{ref_CST_AndrewZhang}, there is a promising opportunity to leverage BSs for countering unauthorized UAVs. In addition to providing communication and sensing services, BSs can also offer effective countermeasures, ensuring the lawful and orderly use of the airspace while minimizing the need for costly specialized equipment.

\subsection{Prior Works}
Most of existing ISAC studies focus on utilizing a unified hardware platform and/or waveform to simultaneously achieve demodulation of communication data and extraction of target parameters with noisy observations\cite{ref_TCOM_FanLiu,ref_JSAC_22FanLiu, ref_JSTSP_AndrewZhang}.
User's activity, channel, and location in extra-large MIMO systems were joint sensed in \cite{ref_2024_LiQiao}.
A compressed sampling perspective to facilitate ISAC processing was proposed in \cite{ref_TWC_ZiweiWan}.
User equipment (UE)’s uplink channel and location in Terahertz extra-large array systems were jointly obtained in \cite{ref_LZR}.
Sensing-centric index modulation-based ISAC systems was proposed in \cite{ref_JSTSP_DingyouMa_IM}.
Given the potential of UAVs to empower a wide range of industries, there is also an abundance of work integrating UAVs with ISAC\cite{ref_CM_UAV_ISAV_Swarm,ref_WC_UAV_ISAC_KaitaoMeng,ref_arXiv_IllegalUAV_NoJamming}.

However, the aforementioned ISAC efforts fall short in considering countermeasures against unauthorized UEs.
Various methods for both jamming and anti-jamming have been extensively reviewed in \cite{refCST}.
To address the challenges posed by interference during communication, specialized waveforms were developed to maximize detection capabilities \cite{refTSPISCJWaveformDesign}. 
Additionally, researchers have investigated the use of jamming signal echoes for sensing, while simultaneously preserving communication functionality. This approach enables dynamic resource allocation based on the sensing outcomes \cite{refICASSPSafeguardingUAVISJC}.
In \cite{refTAESearlyaccessAdaptiveJammingFHSS}, deep reinforcement learning (DRL) methods were employed to jam devices utilizing frequency hopping spread spectrum, leveraging specialized jamming equipment.
The scenario of a full-duplex system with a single UE and UAV was also investigated, where the unauthorized UAV was detected and jammed with the uplink signals from the UE \cite{refCLXinyiWang}.
Furthermore, simultaneous communication and jamming in the presence of CSI errors was rigorously studied \cite{refTVTRobustBeamformingICAJ}.
A novel approach using optical devices to generate electromagnetic signals was proposed in \cite{refJournalLightwaveTech}, enabling secure communication while conducting radar range and velocity deception jamming. 
Real-world tests have demonstrated the efficacy of narrowband signals with low energy in disrupting the fifth-generation networks \cite{refAccessRealTest}.

The concept of integrating radar, communications, and jamming was proposed for jamming illegal UEs in \cite{ref_TWC_NanChiSu}.
The premise of this study centers around security, with the assumption that an entity being jammed is an eavesdropper.
This assumption is restrictive and does not necessarily apply to scenarios involving all the unauthorized UAV flights, where the unauthorized UAVs might not engage in eavesdropping.
Therefore, further research is required to overcome the challenge of jamming common unauthorized UAVs.
Recently, a generalized integrated communications and jamming (ICAJ) framework has been proposed to enable communication and jamming to reinforce each other in one system, thereby achieving more reliable communication and more efficient jamming\cite{ref_CM_ICAJ, ref_ChinaComm_ICAJ}.
However, the aforementioned studies assume separate communication data streams and jamming streams.
In reality, the communication data streams and jamming streams can be jointly designed for a more coordinated ICAJ.

\begin{table*}[]
	\vspace{-5mm}
	\caption{State-of-the-art Jamming Solutions} 
	
	\begin{tabular}{c|ccc|ccc|c|c}
			\hline \hline
			\multirow{4}{*}{\textbf{Ref.}} & \multicolumn{3}{c|}{\textbf{Jamming Methods}}                                                                                                                                                                                                                          & \multicolumn{3}{c|}{\textbf{Jamming Categories}}                                & \multirow{4}{*}{\begin{tabular}[c]{@{}c@{}}\textbf{Relationship} \\ \textbf{with Sensing } \\ \textbf{and Communi-}\\ \textbf{cation}\end{tabular}} & \multirow{4}{*}{\textbf{Algorithms}} \\ \cline{2-7}
			& \multicolumn{1}{c|}{\begin{tabular}[c]{@{}c@{}}Additional\\ Equipment\end{tabular}} & \multicolumn{1}{c|}{\begin{tabular}[c]{@{}c@{}}Additional\\ Jamming\\ Streams\end{tabular}} & \begin{tabular}[c]{@{}c@{}}Reusing\\ Communi-\\ cation Signals\end{tabular} & \multicolumn{1}{c|}{Active} & \multicolumn{1}{c|}{Passive} & Deception &                                                                                                         &                            \\  \hline \cite{refWclRisJamming}
			& \multicolumn{1}{c|}{\checkmark}                                                               & \multicolumn{1}{c|}{}                                                                       &        \checkmark                                                                   & \multicolumn{1}{c|}{}       & \multicolumn{1}{c|}{\checkmark}        &           &                           \textcircled{1}                                                                              &    BCD+SDR                        \\ \hline \cite{refIoTJDestructiveRIS}
			& \multicolumn{1}{c|}{\checkmark}                                                               & \multicolumn{1}{c|}{}                                                                       &   \checkmark                                                                        & \multicolumn{1}{c|}{}       & \multicolumn{1}{c|}{\checkmark}        &           &                           \textcircled{1}                                                                              &   {\begin{tabular}[c]{@{}c@{}}alternating opti-\\ mization+SCA\end{tabular}}                          \\ \hline \cite{refTVTillegalRISchannelAging}
			& \multicolumn{1}{c|}{\checkmark}                                                               & \multicolumn{1}{c|}{}                                                                       &      \checkmark                                                                     & \multicolumn{1}{c|}{}       & \multicolumn{1}{c|}{\checkmark}        &           &                           \textcircled{1}                                                                              &       ————                     \\ \hline \cite{refTSPISCJWaveformDesign}
			& \multicolumn{1}{c|}{}                                                               & \multicolumn{1}{c|}{\checkmark}                                                                       &                                                                           & \multicolumn{1}{c|}{\checkmark}       & \multicolumn{1}{c|}{}        &           &                            \textcircled{3}                                                                              &   {\begin{tabular}[c]{@{}c@{}}convex\\ optimization\end{tabular}}                         \\ \hline \cite{refICASSPSafeguardingUAVISJC}
			& \multicolumn{1}{c|}{}                                                               & \multicolumn{1}{c|}{\checkmark}                                                                       &                                                                           & \multicolumn{1}{c|}{\checkmark}       & \multicolumn{1}{c|}{}        &           &                           \textcircled{3}                                                                               &      {\begin{tabular}[c]{@{}c@{}}S-procedure\\ SDR\end{tabular}}                      \\ \hline \cite{refTAESearlyaccessAdaptiveJammingFHSS}
			& \multicolumn{1}{c|}{\checkmark}                                                               & \multicolumn{1}{c|}{}                                                                       &                                                                           & \multicolumn{1}{c|}{\checkmark}       & \multicolumn{1}{c|}{}        &           &                           \textcircled{1}                                                                              &   DRL                         \\ \hline \cite{refCLXinyiWang}
			& \multicolumn{1}{c|}{}                                                               & \multicolumn{1}{c|}{\checkmark}                                                                       &                                                                           & \multicolumn{1}{c|}{\checkmark}       & \multicolumn{1}{c|}{}        &           &                           \textcircled{3}                                                                               &  {\begin{tabular}[c]{@{}c@{}}alternating opti-\\ mization+SCA\end{tabular}}                           \\ \hline \cite{refTVTRobustBeamformingICAJ}
			& \multicolumn{1}{c|}{}                                                               & \multicolumn{1}{c|}{\checkmark}                                                                       &                                                                           & \multicolumn{1}{c|}{\checkmark}       & \multicolumn{1}{c|}{}        &           &                           \textcircled{3}                                                                              &    {\begin{tabular}[c]{@{}c@{}}customized\\optimization\end{tabular}}                        \\ \hline \cite{refJournalLightwaveTech}
			& \multicolumn{1}{c|}{\checkmark}                                                               & \multicolumn{1}{c|}{}                                                                       &                                                                           & \multicolumn{1}{c|}{\checkmark}       & \multicolumn{1}{c|}{}        &    \checkmark       &         \textcircled{2}                                                                                                &         ————                   \\ \hline
			\cite{ref_TWC_NanChiSu}
			& \multicolumn{1}{c|}{}                                                               & \multicolumn{1}{c|}{\checkmark}                                                                       &                                                                           & \multicolumn{1}{c|}{\checkmark}       & \multicolumn{1}{c|}{}        &           &                           \textcircled{3}                                                                               &   {\begin{tabular}[c]{@{}c@{}}SDR+convex\\ optimization\end{tabular}}                         \\\hline \cite{ref_CM_ICAJ}
			&  \multicolumn{1}{c|}{}                                                               & \multicolumn{1}{c|}{\checkmark}                                                                       &                                                                           & \multicolumn{1}{c|}{\checkmark}       & \multicolumn{1}{c|}{}        &           &   \textcircled{2}                                                                                                      &   {\begin{tabular}[c]{@{}c@{}}SDR+interior\\point method\end{tabular}}                         \\ \hline \cite{ref_ChinaComm_ICAJ}
			& \multicolumn{1}{c|}{}                                                               & \multicolumn{1}{c|}{\checkmark}                                                                       &                                                                           & \multicolumn{1}{c|}{\checkmark}       & \multicolumn{1}{c|}{}        &           &                           \textcircled{3}                                                                              &  DRL                          \\ 
			\hline
			proposed
			& \multicolumn{1}{c|}{}                                                               & \multicolumn{1}{c|}{}                                                                       &   \checkmark                                                                        & \multicolumn{1}{c|}{\checkmark}       & \multicolumn{1}{c|}{}        &           &                           \textcircled{2}                                                                               &   {\begin{tabular}[c]{@{}c@{}}SDR+convex\\ optimization\end{tabular}}                         \\ 
			\hline
			\hline
		\end{tabular}
	\label{table_jamming_work}
	\begin{itemize}
			\item \textcircled{1}: jamming, \textcircled{2}: jamming+communication, \textcircled{3}: jamming+communication+sensing.
			
			\item SDR: semi-definite relaxation, DRL: deep reinforcement learning, BCD: block coordinate descent, SCA: successive convex approximation.
		\end{itemize}
\end{table*}

In addition to these active jamming techniques, passive methods have also garnered attention. 
When acting as a collaborator, an RIS can effectively cover blind spots and enhance signal quality.
However, when an RIS operates as an interferer, it can reshape the characteristics of multipath signals passing through it, resulting in destructive interference at the receiver \cite{refWclRisJamming, refIoTJDestructiveRIS, refTVTillegalRISchannelAging}.
The aforementioned jamming-related work is summarized in Table \ref{table_jamming_work}.

Parallel to these advancements, there has been significant progress in anti-jamming research. 
Authors in \cite{refAntiJammingearlyaccess} examined the dynamic interaction between a jammer and a BS, where the BS continuously adjusts its beamformer in response to the jamming power of the jammer. 
Authors in \cite{refTVTAntiJamming} focused on designing robust passive beamforming using RIS to counter jamming, even when angle information is inaccurate.
Furthermore, a different approach was proposed to utilize additional transceivers in the environment to remodulate jamming signals, embedding structured information for potential use \cite{refTWCAntiJamming}.
The continuous advancements in jamming and anti-jamming technologies reinforce each other, driving ongoing progress in both fields.

Essentially, integrating ISAC with UAV countermeasures relies on downlink beamforming as one of the key technologies.
There are numerous classical schemes.
For fully digital beamforming, there are famous channel inversion (CI)-based scheme\cite{ref_03CM_MultiUserDnBF}, linear minimum mean squared error-based scheme\cite{ref_TCOM_LMMSE}, and optimization-based schemes \cite{ref_10SPM_Ottersten,ref_99_Ottersten, ref_04TVT_DnBFOpt}.
For hybrid beamforming, there are semi-definite programming (SDP)-based schemes\cite{ref_19TWC_QingqingWu_SDP},
orthogonal matching pursuit-based schemes\cite{ref_TWC_OMP_Precoding},
and alternating minimization-based schemes\cite{ref_JSTSP_XianghaoYu}.
Thanks to the development of deep learning, there are some advanced data-driven-based beamforming schemes \cite{refJSTSPShicongLiu, refJSACMinghuiWu,refarxivMinghuiWu}
However, to the best of our knowledge, in the scenario of integrated communication and unauthorized UAV countermeasures, there is currently no generalized downlink beamforming solution.

\subsection{Our Contributions}

Our contributions are summarized as follows.
\begin{itemize}	
	\item 
	We develop a joint communication and jamming (JCJ) model and formulate an associated optimization problem, which can demonstrate how a BS fulfills essential communication demands while efficiently counters unauthorized UAVs with minimal power expenditure.
	Through theoretical analysis of the formulated problem, we yield the suggestion that additional jamming streams are unnecessary.
	Effective countermeasures against unauthorized UAVs can be seamlessly integrated within the UEs' data streams.
	
	\item 
	Upon further in-depth analysis of the proposed JCJ optimization problem's structure, we develop an additional constraint.
	Specifically, the real part of the sum of the elements along any diagonal direction parallel to the main diagonal of the complex symmetric matrix, including cyclically wrapped diagonals, should be no less than a fraction, ranging between 0 and 1, of the matrix's trace.
	This constraint ensures that the solution obtained through the SDP satisfies the rank-1 constraint.
	Simulations confirm the effectiveness of this constraint. 
	
	\item 
	Since the proposed JCJ scheme does not require additional signal streams to jam unauthorized UAVs, the proposed JCJ scheme can operate effectively when the total number of UEs{\footnote{Here, `UEs' also includes legitimate UAVs.}} and unauthorized UAVs exceeds the number of transmit antennas, where traditional schemes fail to operate altogether.
	Moreover, when the number of transmit antennas is fewer than 32, the proposed JCJ scheme is always superior to the traditional channel inversion scheme.
	
\end{itemize}

\subsection{Notation}
Throughout this paper, if not otherwise stated, normal-face letters, boldface lower letters, and boldface uppercase letters denote scalar variables, column vectors, and matrices, respectively;
the transpose operator and the conjugate transpose operator are denoted by $(\cdot)^T$ and $(\cdot)^H$, respectively;
$j=\sqrt{-1}$ is the imaginary unit;
$\mathbb{R}$ and $\mathbb{C}$ are the sets of real-valued and complex-valued numbers, respectively;
$\mathbf{X}[:,m_1:m_2]$ is the matrix composed of column vectors from $m_1$-th column to $m_2$-th column of matrix $\mathbf{X} \in \mathbb{C}^{N\times M}$;
%$[m]$ in ${\bf{X}}[m](\theta)$ means extracting some elements of ${\bf{X}}$ indexed by $[m]$, where $\theta$ is the argument of ${\bf{X}}[m](\theta)$;
$|s|$ is the magnitude of $s$, whether $s$ is a real number or a complex-valued number;
$\mathcal{CN}(\mu, \sigma^2)$ is the circularly-symmetric complex Gaussian distribution with mean $\mu$ and covariance $\sigma^2$;
$\mathcal{U}(a, b)$ is the uniform distribution with interval $a$ and $b$;
$\mathbb{E}[\cdot]$ is the expectation operation;
$\text{Re}(\cdot)$ denotes extraction of the real part;
$\mathbf{0}_N$ is the $N\times 1$ vector with all the elements being 0;
$\mathbf{0}_{N\times M}$ is the $N\times M$ matrix with all the elements being 0;
$\mathbf{I}_N$ is the $N\times N$ unit matrix;
$\mathbf{F}\succeq \mathbf{0}$ denotes $\mathbf{F}$ is a semi-definite matrix;
$c$ is the speed of light.

\section{System Model and Problem Formulation}\label{sec_SysModel_ProblemFormulation}
In this section, we first present the system model and channel model of the investigated MIMO system, and then formulate the JCJ problem.
\subsection{Downlink Transmission Model}
\begin{figure}[!t]
%	\vspace{-5mm}
	\centering
	\color{black}
	%		\vspace{-3mm}	
	\includegraphics[width=3.4in]{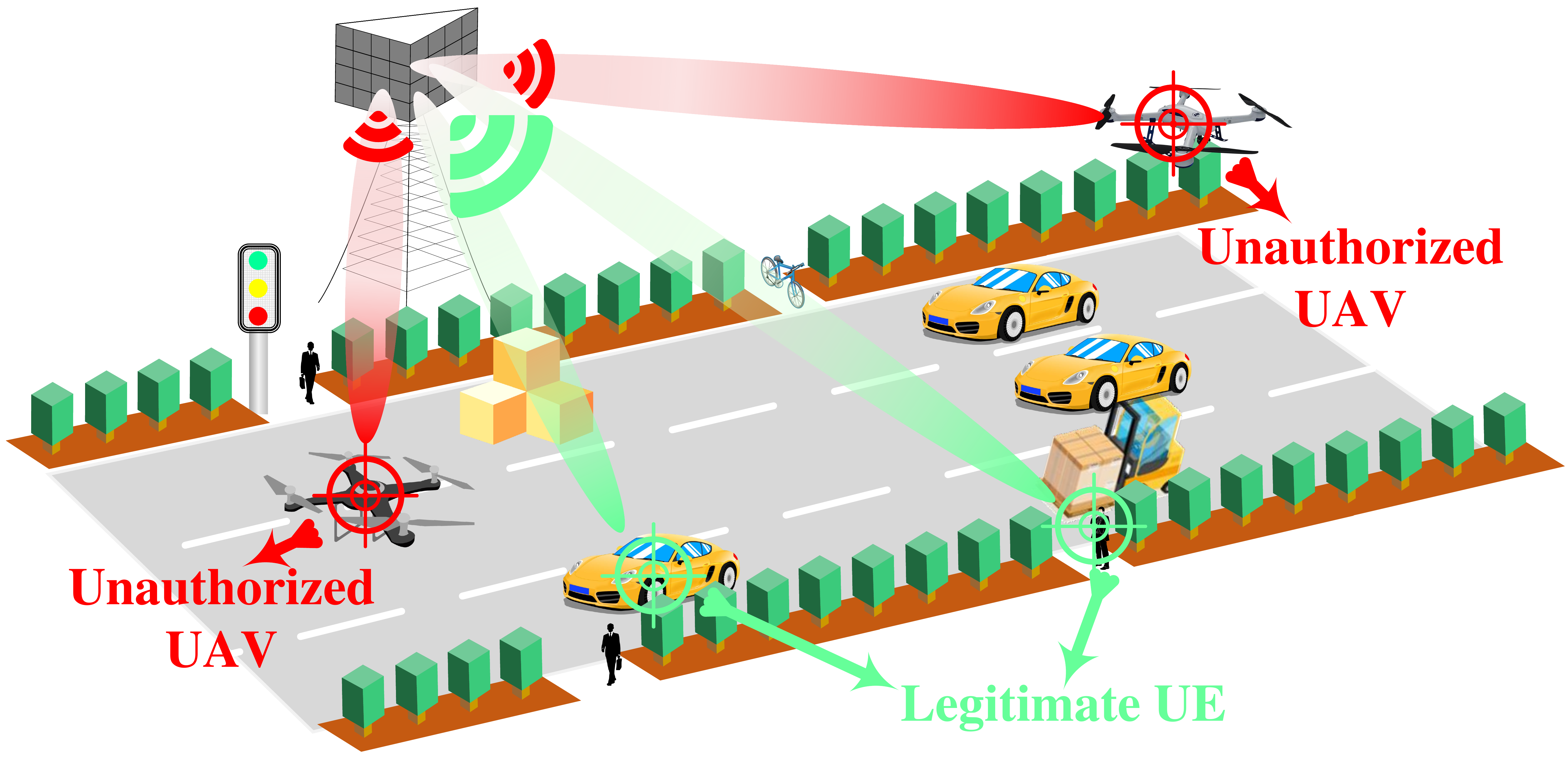}
%	\vspace{-3mm}
	\caption{The considered joint communication and jamming scenario.}
	\label{fig_scene}
%	\vspace{-3mm}
\end{figure}
Consider a multi-user downlink MIMO system, we make the following settings for highlighting the essential aspects of JCJ.
We assume that unauthorized UAVs' operating frequency band is the same as the UE's communication band
{\footnote{We will consider the more general case where unauthorized UAVs' operating band differs from UEs' communication band in future work.}}.
We assume a narrow-band transmission scenario, and a BS adopts a fully digital beamforming architecture\footnote{
For hybrid beamforming architecture, we can derive the hybrid beamformer through the fully digital beamformer obtained in\cite{ref_JSTSP_XianghaoYu,ref_TSP_HybridPrecoding}.
}.
$N_{\text{ue}}$ UEs with a single antenna are served by the BS,
and there are $N_{\text{uav}}$ unauthorized UAVs with a single antenna that should be countered simultaneously.
The BS adopts a uniform linear array equipped with $N_{\text{tx}}$ antennas.
The scenario diagram of the JCJ is shown in Fig. \ref{fig_scene}.
$\mathbf{x}\in \mathbb{C}^{N_{\text{tx}}}$ is the transmitted signal by the BS and can be written as
\begin{align}
	\mathbf{x}=\mathbf{F}\mathbf{s},
\end{align}
where $\mathbf{F}\in \mathbb{C}^{N_{\text{tx}}\times N_s}$ is the beamformer,
${\mathbf{s}} = {[{\mathbf{s}}_{{\text{ue}}}^H\;{\bf{s}}_{{\text{uav}}}^H]^H} \in {^{{N_s}}}$ is the complete data stream and $\mathbb{E}(\mathbf{s}\mathbf{s}^H)=\mathbf{I}_{{N_s}}$,
${\bf{s}}_{{\text{ue}}}\in \mathbb{C}^{N_{\text{ue}}}$ is the UEs' data stream,
${\bf{s}}_{{\text{uav}}}\in \mathbb{C}^{N_{\text{uav}}}$ is the additional jamming stream dedicated to countering unauthorized UAVs,
and $N_s=N_{\text{ue}}+N_{\text{uav}}$.
The aforementioned model has additional $N_{\text{uav}}$ independent data streams to counter $N_{\text{uav}}$ unauthorized UAVs.
However, in practice, jamming unauthorized UAVs does not necessarily require $N_{\text{uav}}$ additional jamming streams, as will be demonstrated by the proposed JCJ scheme and the simulations presented in the following sections.
If some or all of the $N_{\text{uav}}$ additional jamming streams are not required, the corresponding columns of $\mathbf{F}$ can be set to zero.

We consider that ${y_{{\text{ue}},n}}\in \mathbb{C}$ and ${y_{{\text{uav}},m}}\in \mathbb{C}$ are the received signals of the $n$-th UE and the $m$-th unauthorized UAV, respectively, 
and they can be denoted as
\begin{align}
	{y_{{\text{ue}},n}} &= {{\mathbf{h}}^H_{{\text{ue}},n}}{\mathbf{x}} + {w_{{\text{ue}},n},}\nonumber\\
	{y_{{\text{uav}},m}} &= {{\mathbf{h}}^H_{{\text{uav}},m}}{\mathbf{x}} + {w_{{\text{uav}},m}},
\end{align}
where $\mathbf{h}_{\text{ue},n}\in\mathbb{C}^{N_{\text{tx}}}$ and $\mathbf{h}_{\text{uav},m}\in\mathbb{C}^{N_{\text{tx}}}$ are the channels from the BS to the $n$-th UE and from the BS to the $m$-th unauthorized UAV, respectively.
${w_{{\text{ue}},n}}\sim{\cal C}{\cal N}(0,\sigma _{{\text{ue}},n}^2)$ and ${w_{{\text{uav}},n}}\sim{\cal C}{\cal N}(0,\sigma _{{\text{uav}},m}^2)$ are the noise in the UE's and unauthorized UAV's receiver, respectively.
$\sigma _{{\text{ue}},n}^2$ and $\sigma _{{\text{uav}},m}^2$ are the variances of respective noises.
For simplicity and without loss of generality, we model the channel as 
\begin{align}
	\mathbf{h}(r,\theta)=\alpha(r) \mathbf{b}_{N_{\text{tx}}}(\theta),
\end{align}
where $\mathbf{b}_{N_{\text{tx}}}(\theta)=[1, e^{-j2\pi\frac{d}{\lambda}\sin \theta},\ldots,e^{-j2\pi\frac{(N_{\text{tx}}-1)d}{\lambda}\sin \theta}]^T\in\mathbb{C}^{N_{\text{tx}}}$ is the steering vector,
and $\alpha(r) = e^{j\xi}\sqrt{G_{\text{tx}}G_{\text{rx}}\lambda^{2}}/4\pi r$ is the large-scale fading coefficient\cite{ref_LZR}.
$\lambda$ is the wavelength associated with the carrier frequency,
$\theta$ and $r$ are the angle of departure (AoD) and distance between the UE (UAV) and the BS, respectively,
$d=\lambda/2$ is the antenna element spacing,
$G_\text{tx}$ and $G_\text{rx}$ are the antenna gains of the transmit antenna and the receive antenna, respectively,
$\xi\sim\mathcal{U}(0,2\pi)$ is the random phase of the channels.

Then, we can combine the received signals of all the legitimate UEs and unauthorized UAVs into a more compact form as
\begin{align}
	\mathbf{y}& = \mathbf{H}\mathbf{x}+\mathbf{w}\nonumber\\
	& = \mathbf{H} \mathbf{F} \mathbf{s}+\mathbf{w},
\end{align}
where 
${\bf{H}} = {[{\bf{H}}_{{\text{ue}}}\;{\bf{H}}_{{\text{uav}}}]^H}\in\mathbb{C}^{N_s\times N_{\text{tx}}}$,
${\bf{H}}_{{\text{ue}}}=[{\bf{h}}_{{\text{ue},1}}\ ,\ldots,$
${\bf{h}}_{{\text{ue}},N_{\text{ue}}}]\in\mathbb{C}^{N_{\text{tx}}\times N_\text{ue}}$,
${\bf{H}}_{{\text{uav}}}=[{\bf{h}}_{{\text{uav}},1}\ ,\ldots,{\bf{h}}_{{\text{uav}},N_{\text{uav}}}]\in\mathbb{C}^{N_{\text{tx}}\times N_\text{uav}}$
,
${\bf{w}} = [{w_{{\text{ue}},1}},\ldots,w_{{\text{ue}},{N_{{\text{ue}}}}},$
${w_{{\text{uav}},1}},\ldots,{w_{{\text{uav}},{N_{{\text{uav}}}}}}]^T\in\mathbb{C}^{N_s}$,
${\bf{y}} = [{y_{{\text{ue}},1}},\ldots,y_{{\text{ue}},{N_{{\text{ue}}}}},{y_{{\text{uav}},1}},\ldots,$
${y_{{\text{uav}},{N_{{\text{uav}}}}}}]^T\in\mathbb{C}^{N_s}$.

\subsection{Problem Formulation}
The achievable rate of the $n$-th UE can be obtained as
\begin{align}
	&{R_n} \nonumber\\
	&= {\log _2}\left( {1 + \frac{{{{\bf{h}}_{{\text{ue,}}n}^H}{{\bf{f}}_{{\text{ue,}}n}}{\bf{f}}_{{\text{ue,}}n}^H{\bf{h}}_{{\text{ue,}}n}}}{{\sigma _{{\text{ue,}}n}^2 + {{\bf{h}}_{{\text{ue,}}n}^H}{\bf{F}}{{\bf{F}}^H}{\bf{h}}_{{\text{ue,}}n} - {{\bf{h}}_{{\text{ue,}}n}^H}{{\bf{f}}_{{\text{ue,}}n}}{\bf{f}}_{{\text{ue,}}n}^H{\bf{h}}_{{\text{ue,}}n}}}} \right),
\end{align}
where ${{\bf{f}}_{{\text{ue,}}n}}={\bf{F}}[:,n]$.
The SINR of the $m$-th unauthorized UAV can be computed as
\begin{align}
	{\Gamma _{m}} = 10\log_{10} \frac{{{P_{{\text{e}},m}}}}{{{{\bf{h}}_{{\text{uav}},m}^H}{\bf{FF}}_{}^H{\bf{h}}_{{\text{uav}},m} + \sigma _{{\text{uav}},m}^2}},
\end{align}
where ${P_{{\text{e}},m}}$ is the received power of the $m$-th unauthorized UAV from its associated controller.
We assume that BS knows the perfect $\mathbf{H}$.
Then, the JCJ problem can be formulated as
\begin{align}\label{eq_P1}
	\mathcal{P}1:
	\begin{array}{l}
		\arg {\mathop{\min}\limits_{\bf{F}}} \qquad\qquad\quad {P_\text{t}}\nonumber\\
		\qquad{\text{s.t.}}\quad \;C_{1,1}:{R}_n \ge {R_{{\text{th}},n}},\forall n\in1,\ldots,N_{\text{ue}},\nonumber\\
		\qquad\ \;\quad\;\;\;C_{1,2}:{\Gamma_{m}} \le \Gamma_{{\text{th}},m} ,\forall m\in1,\ldots,N_{\text{uav}},	
	\end{array}
\end{align}
where 
${P_\text{t}}=\text{tr}(\mathbf{F}\mathbf{F}^H)$ is the total transmit power,
${R_{{\text{th}},n}}$ and $\Gamma_{{\text{th}},m}$ are the expected achievable rate threshold of the $n$-th UE and the expected SINR threshold of the $m$-th unauthorized UAV, respectively.

\section{Proposed Joint Communication and Jamming Scheme}
\subsection{Problem Transformation}
Since the constraints of problem $\mathcal{P}1$ are quadratic, our objective is to convert $\mathcal{P}1$ into an SDP form.
Constraint $C_{1,1}$ can be expressed as
\begin{align}
\sigma _{{\text{ue}},n}^2({2^{{R_n}}} - 1)/{2^{{R_n}}} + {{\bf{h}}_{{\text{ue}},n}^H}{\bf{F}}{{\bf{F}}^H}{\bf{h}}_{{\text{ue}},n}({2^{{R_n}}} - 1)/{2^{{R_n}}} \nonumber \\ \le {{\bf{h}}_{{\text{ue}},n}^H}{{\bf{f}}_{{\text{ue}},n}}{\bf{f}}_{{\text{ue}},n}{\bf{h}}_{{\text{ue}},n}^H,
\end{align}
where we omit $\forall n\in1,\ldots,N_{\text{ue}}$ for simplicity, and the same applies to the subsequent formulation.
Similarly, $C_{1,2}$ can be expressed as
\begin{align}
	{{{\bf{h}}_{{\text{uav}},m}^H}{\bf{FF}}_{}^H{\bf{h}}_{{\text{uav}},m} \ge {P_{{\text{e}},m}}{{10}^{ - {\Gamma _{{\text{th}},m}}/10}} - \sigma _{{\text{uav}},m}^2}.
\end{align}
By defining
${\bf{f}} = {[{\bf{f}}_{{\text{ue}},1}^H\; \cdots \;{\bf{f}}_{{\text{ue}},{N_{{\text{ue}}}}}^H\;{\bf{f}}_{{\text{uav}},1}^H\; \cdots \;{\bf{f}}_{{\text{uav}},{N_{{\text{uav}}}}}^H]^H}\in \mathbb{C}^{N_sN_{\text{tx}}}$,
${{\bf{f}}_{{\text{uav,}}n}}={\bf{F}}[:,n+N_{\text{ue}}]$,
${\bf{\tilde F}} = {\bf{f}}{{\bf{f}}^H}\in \mathbb{C}^{N_sN_{\text{tx}}\times N_sN_{\text{tx}}}$,
\begin{align}
		{\bf{\tilde H}}_{1,n}^{} &= \left[ {\begin{array}{*{20}{c}}
				{{{\bf{h}}_{{\text{ue}},n}^H}}&{\bf{0}}& \cdots &{\bf{0}}\\
				{\bf{0}}& \ddots & \ddots & \vdots \\
				\vdots & \ddots &{{{\bf{h}}_{{\text{ue}},n}^H}}&{\bf{0}}\\
				{\bf{0}}& \cdots &{\bf{0}}&{{{\bf{h}}_{{\text{ue}},n}^H}}
		\end{array}} \right]\in \mathbb{C}^{N_s\times N_sN_{\text{tx}}},\nonumber\\
		{\bf{\tilde H}}_{2,n}^{} &= \left[ {\begin{array}{*{20}{c}}
				{{{\bf{0}}_{1 \times (n - 1){N_{{\text{tx}}}}}}}&{{{\bf{h}}_{{\text{ue}},n}^H}}&{{{\bf{0}}_{1 \times ({N_{{\text{ue}}}} - n + {N_{{\text{uav}}}}){N_{{\text{tx}}}}}}}
		\end{array}} \right]\nonumber\\ &\qquad\qquad\qquad\qquad\qquad\qquad\qquad\in {\mathbb{C}^{1 \times {N_s}{N_{{\text{tx}}}}}},\nonumber\\
		{\bf{\bar H}}_m^{} &= \left[ {\begin{array}{*{20}{c}}
				{{{\bf{h}}_{{\text{uav}},m}^H}}&{\bf{0}}& \cdots &{\bf{0}}\\
				{\bf{0}}& \ddots & \ddots & \vdots \\
				\vdots & \ddots &{{{\bf{h}}_{{\text{uav}},m}^H}}&{\bf{0}}\\
				{\bf{0}}& \cdots &{\bf{0}}&{{{\bf{h}}_{{\text{uav}},m}^H}}
		\end{array}} \right]\in \mathbb{C}^{N_s\times N_sN_{\text{tx}}},\nonumber
\end{align}
we can transform the problem $\mathcal{P}1$ equivalently to the problem $\mathcal{P}2$ as
\begin{align}\label{eq_P2}
	\mathcal{P}2:
	\begin{array}{*{20}{l}}
		{\arg \mathop {\min }\limits_{{\bf{\tilde F}}}\qquad  {\text{tr(}}{\bf{\tilde F}})}\\
		 {\text{s.t.}} \;\; \;C_{2,1}:{\text{tr}}({{\bf{A}}_{1,n}}{\bf{\tilde F}}) \ge {\gamma _{1,n}},\forall n=1,\ldots,N_{\text{ue}},\\
		{\quad \quad C_{2,2}:{\text{tr}}({{\bf{A}}_{2,m}}{\bf{\tilde F}}) \ge {\gamma _{2,m}}},\forall m=1,\ldots,N_{\text{uav}},\nonumber\\
		\qquad\qquad\quad{\bf{\tilde F}}\succeq \mathbf{0},\;{\text{rank(}}{\bf{\tilde F}}{\text{) = }}1,
	\end{array}
\end{align}
where
\begin{align}	 
		{{\bf{A}}_{1,n}} &= {\bf{\tilde H}}_{2,n}^H{\bf{\tilde H}}_{2,n}^{} - ({2^{{R_{{\text{th}},n}}}} - 1)/{2^{{R_{{\text{th}},n}}}}{\bf{\tilde H}}_{1,n}^H{\bf{\tilde H}}_{1,n}^{}\nonumber\\
		&\qquad\qquad\qquad\qquad\qquad\in \mathbb{C}{^{{N_s}{N_{{\text{tx}}}} \times {N_s}{N_{{\text{tx}}}}}},\\
		{{\bf{A}}_{2,m}} &= {\bf{\bar H}}_m^H{\bf{\bar H}}_m^{} \in \mathbb{C}{^{{N_s}{N_{{\text{tx}}}} \times {N_s}{N_{{\text{tx}}}}}},\\
		{\gamma _{1,n}} &= ({2^{{R_{{\text{th}},n}}}} - 1)/{2^{{R_{{\text{th}},n}}}}\sigma _{{\text{ue}},n}^2,\\
		{\gamma _{2,m}} &= {P_{{\text{e}},m}}{10^{ - {\Gamma _{{\text{th}},m}}/10}} - \sigma _{{\text{uav}},m}^2.
\end{align}

If the constraint ${\text{rank(}}{\bf{\tilde F}}{\text{) = }}1$ is removed, the problem $\mathcal{P}2$ is a standard SDP problem and can be solved by typical convex optimization tools, e.g., CVX\cite{ref_cvx}.
Therefore, we resort to semi-definite relaxation (SDR) and drop the rank constraint to obtain a relaxed version of $\mathcal{P}2$ as
\begin{align}\label{eq_P3}
	\mathcal{P}3:
	\begin{array}{*{20}{l}}
		{\arg \mathop {\min }\limits_{{\bf{\tilde F}}}\qquad  {\text{tr(}}{\bf{\tilde F}})}\\
		{\text{s.t.}} \;\; \;C_{3,1}:{\text{tr}}({{\bf{A}}_{1,n}}{\bf{\tilde F}}) \ge {\gamma _{1,n}},\forall n=1,\ldots,N_{\text{ue}},\\
		{\quad \quad C_{3,2}:{\text{tr}}({{\bf{A}}_{2,m}}{\bf{\tilde F}}) \ge {\gamma _{2,m}}},\forall m=1,\ldots,N_{\text{uav}},\nonumber\\
		\qquad\qquad\quad{\bf{\tilde F}}\succeq \mathbf{0}.
	\end{array}
\end{align}
Note that the power requirement for $\mathcal{P}3$ is less than or equal to that of $\mathcal{P}2$, since $\mathcal{P}2$ has an additional rank-1 constraint compared to $\mathcal{P}3$.
Although $\mathcal{P}3$ can be solved by convex optimization methods, the solution to $\mathcal{P}3$ often fails to satisfy the rank-1 constraint.
Consequently, transforming the solution to $\mathcal{P}3$ into the solution to $\mathcal{P}2$ poses a challenge.
The conventional approach entails applying eigenvalue decomposition to the solution to $\mathcal{P}3$ as described in \cite{ref_JSTSP_XianghaoYu}.
The eigenvector corresponding to the largest eigenvalue is then selected as the solution to $\mathcal{P}2$.
However, Corollary \ref{corollary_1} demonstrates that eigenvalue decomposition is not suitable for $\mathcal{P}3$.
First, by analyzing the structure of $\mathcal{P}3$, we can obtain Theorem \ref{theorem_1}, the proof of which is given in Appendix \ref{appendix_1}.
\begin{Theorem}\label{theorem_1}
	If $N_{\text{ue}}\ge 1$, the solution to $\mathcal{P}3$ satisfies ${\bf{\tilde F}} = {\text{blkdiag}}({{\bf{S}}_1},\ldots,{{\bf{S}}_{{N_{{\text{ue}}}}}},{\bf{0}}_{N_{\text{uav}}N_{\text{tx}}\times N_{\text{uav}}N_{\text{tx}}})$,
	where $\mathbf{S}_{n},n=1,\ldots,N_\text{ue},$ is an $N_\text{tx}\times N_\text{tx}$ symmetric matrix.
\end{Theorem}
Corollary \ref{corollary_1} demonstrates that the solution to $\mathcal{P}3$ is definitely not the solution to $\mathcal{P}2$ when the number of UEs exceeds one.
\begin{Corollary}\label{corollary_1}
	If $N_{\text{ue}}\ge 2$, the solution to $\mathcal{P}2$ cannot be obtained by eigenvalue decomposition of the solution to $\mathcal{P}3$.
\end{Corollary}
\begin{proof}
	The rank of block-diagonal matrix is larger than or equal to the number of block.
	Since $N_{\text{ue}}\ge 2$, the rank of the solution to $\mathcal{P}3$ is larger than or equal to 2.
	However, the rank of the solution to $\mathcal{P}2$ needs to be 1.
	Therefore, we cannot obtain the solution to $\mathcal{P}2$ by computing the eigenvalue decomposition of the solution to $\mathcal{P}3$.
	That is, the eigenvector corresponding to the maximum eigenvalue of the solution to $\mathcal{P}3$ cannot form the solution to $\mathcal{P}2$.
	\qedhere
\end{proof}
Although $\mathcal{P}3$ can be solved easily by standard convex optimization, the discrepancy between $\mathcal{P}3$ and $\mathcal{P}2$ compels us to seek alternative methods.
By further analyzing the structure of $\mathcal{P}2$, we can obtain Theorem \ref{theorem_2}, whose proof is given in Appendix \ref{appendix_2}.

\begin{Theorem}\label{theorem_2}
	If $N_{\text{ue}}\ge 1$, the solution to $\mathcal{P}2$ satisfies ${\bf{\tilde F}} = {\text{blkdiag}}({\bf{\tilde S}},{\bf{0}}_{N_{\text{uav}}N_{\text{tx}}\times N_{\text{uav}}N_{\text{tx}}})$, where $\mathbf{\tilde S}$ is an $N_{\text{ue}}N_\text{tx}\times N_{\text{ue}}N_\text{tx}$ symmetric matrix.
\end{Theorem}
Drawing from Theorem \ref{theorem_2}, we can directly obtain Corollary \ref{corollary_2} as follows.
\begin{Corollary}\label{corollary_2}
	If $N_{\text{ue}}\ge 1$, additional jamming streams are not required for jamming unauthorized UAVs.
\end{Corollary}
\begin{proof}
	Notice the constraint $\text{rank}(\mathbf{\tilde F})=1$ and $\mathbf{\tilde F}=\mathbf{\tilde f}\mathbf{\tilde f}^H$, where $\mathbf{\tilde f}\in \mathbb{C}^{N_sN_{\text{tx}}}$.
	Through Theorem \ref{theorem_2}, we know that ${\bf{\tilde F}}[{N_{{\text{ue}}}}{N_{{\text{tx}}}} + 1:{N_s}{N_{{\text{tx}}}},\;{N_{{\text{ue}}}}{N_{{\text{tx}}}} + 1:{N_s}{N_{{\text{tx}}}}] = {\bf{0}}_{N_{\text{uav}}N_{\text{tx}}\times N_{\text{uav}}N_{\text{tx}}}$.
	This means ${\bf{\tilde f}}[{N_{{\text{ue}}}}{N_{{\text{tx}}}} + 1:{N_s}{N_{{\text{tx}}}}] = {\bf{0}}_{N_{\text{uav}}N_{\text{tx}}}$, which indicates that additional jamming streams are not required.
	\qedhere
\end{proof}

By comparing Theorem \ref{theorem_1} and Theorem \ref{theorem_2}, we know that the solution to $\mathcal{P}2$ cannot be accurately obtained by the solution to $\mathcal{P}3$.
This implies that the SDR is loose.
Therefore, we should formulate a new problem to solve problem $\mathcal{P}2$.

\subsection{Problem Reformulation}
By comparing Theorem \ref{theorem_1} and Theorem \ref{theorem_2}, we know that SDR makes ${\bf{\tilde F}}[1:{N_{{\text{ue}}}}{N_{{\text{tx}}}},\;1:{N_{{\text{ue}}}}{N_{{\text{tx}}}}]$ block diagonalized.
If we consider using SDP to solve $\mathcal{P}2$, we can add some additional constraints on top of SDR to ensure that ${\bf{\tilde F}}[1:{N_{{\text{ue}}}}{N_{{\text{tx}}}},\;1:{N_{{\text{ue}}}}{N_{{\text{tx}}}}]$ is no longer block-diagonal.

To this end, we define $\mathbf{\bar P}_k \in \mathbb{C}^{N_sN_{\text{tx}}\times N_sN_{\text{tx}}},\; k = 1,\ldots,{N_{{\text{ue}}}}{N_{\text{tx}}} - 1,$ as
\begin{align}
	\mathbf{\bar P}_k = \left[ {\begin{array}{*{20}{c}}
			{{\text{circshift}}({{\bf{I}}_{{N_{{\text{ue}}}}{N_{{\text{tx}}}}}},k)}&{{{\bf{0}}_{{N_{{\text{ue}}}}{N_{{\text{tx}}}} \times {N_{{\text{uav}}}}{N_{{\text{tx}}}}}}}\\
			{{{\bf{0}}_{{N_{{\text{uav}}}}{N_{{\text{tx}}}} \times {N_{{\text{ue}}}}{N_{{\text{tx}}}}}}}&{{{\bf{0}}_{{N_{{\text{uav}}}}{N_{{\text{tx}}}}\times {N_{{\text{uav}}}}{N_{{\text{tx}}}}}}}
	\end{array}} \right],
\end{align}
where ${\text{circshift}}({{\bf{I}},k)}$ denotes a cyclic rightward shift of an identity matrix $\bf{I}$ by $k$ units.
Given ${\bf{\tilde F}} = {\bf{f}}{{\bf{f}}^H}\in \mathbb{C}^{N_sN_{\text{tx}}\times N_sN_{\text{tx}}}$, we can use Theorem \ref{theorem_3} to construct an additional constraint.
The proof of Theorem \ref{theorem_3} is given in Appendix \ref{appendix_3}.
\begin{Theorem}\label{theorem_3}
	For any $\mathbf{a}\in \mathbb{C}^{K}$, $\text{Re}(\text{tr}(\mathbf{A}))\ge \text{Re}(\text{tr}(\mathbf{AP}_k)), \forall k=0,\ldots,K-1$, where $\mathbf{A}=\mathbf{a}\mathbf{a}^H\in \mathbb{C}^{K\times K}$ and $\mathbf{P}_k={\text{circshift}}({\bf{ I}}_K,\;k)$.
\end{Theorem}
By using Theorem \ref{theorem_3}, to ensure that ${\bf{\tilde F}}[1:{N_{{\text{ue}}}}{N_{{\text{tx}}}},\;1:{N_{{\text{ue}}}}{N_{{\text{tx}}}}]$ is no longer block-diagonal, we construct the additional constraint as
\begin{align}
	\text{Re}({\text{tr}}({\bf{\tilde F\bar P}}_k^{})) \ge \eta \text{Re}({\text{tr}}({\bf{\tilde F}})), \forall k=1,\ldots,N_{\text{ue}}N_{\text{tx}}-1,
\end{align}
where $\eta \in (0,1)$ is a parameter whose value is associated with $\mathbf{H}$.
The motivation behind this constraint is to make ${\bf{\tilde F}}[1:{N_{{\text{ue}}}}{N_{{\text{tx}}}},\;1:{N_{{\text{ue}}}}{N_{{\text{tx}}}}]$ no longer a block-diagonal matrix, thereby allowing ${\bf{\tilde F}}[1:{N_{{\text{ue}}}}{N_{{\text{tx}}}},\;1:{N_{{\text{ue}}}}{N_{{\text{tx}}}}]$ to be factorized into the outer product of two identical vectors.
Finally, the optimization problem can be reformulated as
\begin{algorithm}[!t]
	\label{alg_JCJ}
	\caption{Joint Communication and Jamming Scheme}%算法名字
	\LinesNumbered %要求显示行号
	\KwIn{number of BS antennas $N_{\text{tx}}$, number of UEs $N_{\text{ue}}$, number of unauthorized UAVs $N_{\text{uav}}$, channel $\mathbf{H}$, predefined  $\boldsymbol\Phi\in\mathbb{R}^{V}$, predefined achievable rate threshold for UEs $R_{\text{th},n}, \forall n=1,\ldots,N_{\text{ue}}$, predefined SINR threshold for unauthorized UAVs $\Gamma_{\text{th},m}, \forall m=1,\ldots,N_{\text{uav}}$.
	}%输入参数
	\KwOut{beamformer $\mathbf{\bar F}$
	}%输出
	
	$\text{Error}_{\text{all}} = [\ ]$\;
	$\mathbf{\hat f}_{\text{all}} = [\ ]$\;
	\For{$v=1,2,\ldots,V$}{		$\eta=\boldsymbol\Phi[v]$\;
		Use CVX to solve $\mathcal{P}3'$ and obtain $\mathbf{\tilde F}_v$\;			
		Select the eigenvector $\mathbf{\hat f}_v$ corresponding to the maximum eigenvalue of $\mathbf{\tilde F}_v$\;
		$\mathbf{\hat f}_{\text{all}} = [\mathbf{\hat f}_{\text{all}}\ \mathbf{\hat f}_v]$\;
		Obtain the possible solution to $\mathcal{P}3'$ from (\ref{eq_solution4})\;
		Obtain ture achievable rate and SINR from (\ref{eq_true_R_gamma})\;
		Compute $\text{Error}_v$ from (\ref{eq_error_v})\;
		$\text{Error}_{\text{all}} = [\text{Error}_{\text{all}}\ \text{Error}_v]$			
	}
	Select the index of the best $\eta$ as $\hat v$ from (\ref{eq_best_eta})\;
	Obtain the final beamformer from (\ref{eq_final_precoder})\;
	
\end{algorithm}
\begin{align}\label{eq_P4}
	\mathcal{P}3':
\begin{array}{*{20}{l}}
	{\arg \mathop {\min }\limits_{{\bf{\tilde F}}} \qquad {\text{tr(}}{\bf{\tilde F}})}\\
	{{\text{s.t.}}\ {C_{3,1}}:{\text{tr}}({{\bf{A}}_{1,n}}{\bf{\tilde F}}) \ge {\gamma _{1,n}},\forall n = 1, \ldots ,{N_{{\text{ue}}}}},\\
	\begin{array}{l}
		\quad  {C_{3,2}}:{\text{tr}}({{\bf{A}}_{2,m}}{\bf{\tilde F}}) \ge {\gamma _{2,m}},\forall m = 1, \ldots ,{N_{{\text{uav}}}},\\
		\quad  {C_{3,3}}:\text{Re}({\text{tr}}({{\bf{A}}_{3,k}}{\bf{\tilde F}})) \ge {\gamma _{3,k}},\forall k = 1,\ldots,{N_{{\text{ue}}}}{N_{{\text{tx}}}}-1,\nonumber
	\end{array}\\
	{\qquad \qquad \quad {\bf{\tilde F}}\succeq \mathbf{0}},
\end{array}
\end{align}
where ${{\bf{A}}_{3,k}} = {{\bf{\bar P}}_k} - \eta {\bf{ I}}_{N_sN_{\text{tx}}}$ and ${\gamma _{3,k}}=0$.

Based on $\mathcal{P}3'$, the JCJ scheme is described as follows.
First, we should predefine some possible values of $\eta$ in $\boldsymbol\Phi\in \mathbb{R}^{V}$, since we do not know which $\eta$ is the best to solve $\mathcal{P}3'$\footnote{
Determining a suitable value for $\eta$ based on $\mathbf{H}$ remains an open direction for future investigation.
}
.

Second, for each $\eta$ in ${\boldsymbol\Phi}$, we use CVX to solve $\mathcal{P}3'$ and obtain the corresponding $\mathbf{\tilde F}_v$.
$\mathbf{\tilde F}_v$ is the result corresponding to the $v$-th element in $\boldsymbol\Phi$.
The eigenvector $\mathbf{\hat f}_v\in \mathbb{C}^{N_sN_{\text{tx}}}$ corresponding to the maximum eigenvalue of $\mathbf{\tilde F}_v$ is obtained by eigenvalue decomposition.
All values of $\mathbf{\hat f}_v$ corresponding to ${\boldsymbol\Phi}$ are arranged as 
\begin{align}
	\hat{\mathbf{F}}_{\text{all}}=[\hat{\mathbf{f}}_1, \hat{\mathbf{f}}_2, \ldots, \hat{\mathbf{f}}_V]\in \mathbb{C}^{N_sN_{\text{tx}}\times V}.
\end{align}

Third, we obtain $\mathbf{\hat F}_v \in \mathbb{C}^{N_{\text{tx}}\times N_s}$ as
\begin{align}\label{eq_solution4}
	\mathbf{\hat F}_v[:,i]=\mathbf{\hat f}_v[(i-1)N_{\text{tx}}+1:iN_{\text{tx}}], i=1,\ldots,N_s.
\end{align}

Finally, we choose the $\mathbf{\hat F}_v$ that best satisfies $C_{3,1}$ and $C_{3,2}$ as the result of the proposed JCJ scheme.
Specifically, the true achievable rate and SINR can be obtained as
\begin{align}\label{eq_true_R_gamma}
		{{\hat \gamma }_{1,n,v}} &= {\text{tr}}({{\bf{A}}_{1,n}}{\bf{\hat F}}_v),\;n = 1,\ldots,{N_{{\text{ue}}}},\nonumber\\
		{{\hat \gamma }_{1,m,v}} &= {\text{tr}}({{\bf{A}}_{2,m}}{\bf{\hat F}}_v),\;m = 1,\ldots,{N_{{\text{uav}}}}.
\end{align}
We define the error corresponding to each $\eta$ as
\begin{align}\label{eq_error_v}
	\text{Error}_v=\max(|\gamma_{1,1}-{\hat\gamma}_{1,1,v}|,\ldots,|\gamma_{1,N_{\text{ue}}}-{\hat\gamma}_{1,N_{\text{ue}},v}|,\nonumber\\ |\gamma_{2,1}-{\hat\gamma}_{2,1,v}|,\ldots,|\gamma_{2,N_{\text{uav}}}-{\hat\gamma}_{2,N_{\text{uav}},v}|),
\end{align}
and all errors associated with $\boldsymbol\Phi$ are defined as $\text{Error}_{\text{all}}=[\text{Error}_1, \text{Error}_2, \ldots,\text{Error}_V]^T\in \mathbb{R}^{V}$.
After $\text{Error}_{\text{all}}$ has been obtained, we select the index of the best $\eta$ as
\begin{align}\label{eq_best_eta}
	\hat v=\arg \mathop {\min }\limits_v \ {\text{Erro}}{{\text{r}}_{{\text{all}}}}[v].
\end{align}
Therefore, the final beamformer $\mathbf{\bar F}\in \mathbb{C}^{N_{\text{tx}}\times N_s}$ is obtained as
\begin{align}\label{eq_final_precoder}
	\mathbf{\bar F}[:,i]=\mathbf{\hat F}_{\text{all}}[(i-1)N_{\text{tx}}+1:iN_{\text{tx}}, \hat v], i=1,\ldots,N_s.
\end{align}
The overall JCJ scheme is summarized in Algorithm \ref{alg_JCJ}.
\begin{Remark}
	In our simulations for problem $\mathcal{P}3'$, we adopt equality constraints for $C_{3,1}$ and $C_{3,2}$. 
	This is because, under equality constraints, the solution composed of the eigenvector corresponding to the largest eigenvalue often satisfies the equality constraints.
	In contrast, under inequality constraints, the solution, composed of the eigenvector corresponding to the largest eigenvalue, typically does not satisfy the inequality constraints.
\end{Remark}

\subsection{Complexity and Convergence Analysis}
SDP problems are commonly solved by the interior point method \cite{ref_SDP_interior_point,ref_SDP_SPM_ZhiQuanLuo},
which obtains an $\epsilon$-optimal solution after a sequence of iterations with the given $\epsilon$.
Therefore, the SDP problem $\mathcal{P}3'$ in (\ref{eq_P4}) can be solved with a worst case complexity of $\mathcal{O}({({N_{{\text{ue}}}}{N_{{\text{tx}}}})^{4.5}}\log (1/\epsilon))$\cite{ref_SDP_SPM_ZhiQuanLuo}.
Consequently, the complexity of the proposed JCJ scheme is $\mathcal{O}({({N_{{\text{ue}}}}{N_{{\text{tx}}}})^{4.5}}V\log (1/\epsilon))$.
Note that the complexity is not related to $N_{\text{uav}}$, as Theorem \ref{theorem_2} indicates that no additional dedicated jamming streams are required.

In the simulation, the problem is solved using the SDPT3 solver in CVX.
Considering that SDPT3 can determine the feasibility and convergence of the problem, the proposed method demonstrates strong robustness\cite{reftutuncusolving}.

\section{Simulation Results}\label{sec_Simulation}
\begin{figure*}[!t]
	\vspace{-5mm}
	\centering
	\color{black}
	\subfigure[$\mathbf{\mathord{\buildrel{\lower3pt\hbox{$\scriptscriptstyle\smile$}} \over F} }$: $N_{\text{ue}}=0$ and $N_{\text{uav}}=4$, solution to $\mathcal{P}3$.]{
		\begin{minipage}[t]{0.48\linewidth}
			\centering
			\includegraphics[width=3.3in]{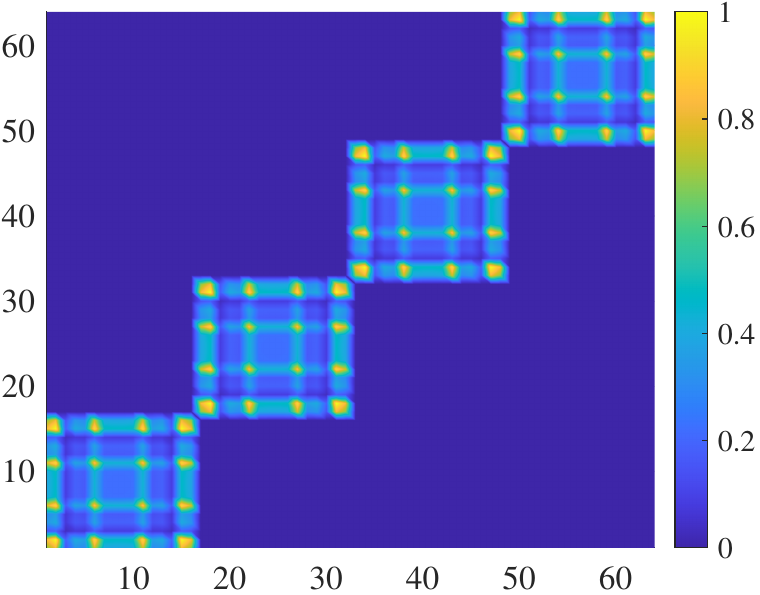}
			\label{fig_UE0_UAV4_Contraint12}
		\end{minipage}
	}
	\subfigure[$\mathbf{\mathord{\buildrel{\lower3pt\hbox{$\scriptscriptstyle\smile$}} \over F} }$: $N_{\text{ue}}=N_{\text{uav}}=2$, solution to $\mathcal{P}3$.]{
		\begin{minipage}[t]{0.48\linewidth}
			\centering
			\includegraphics[width=3.3in]{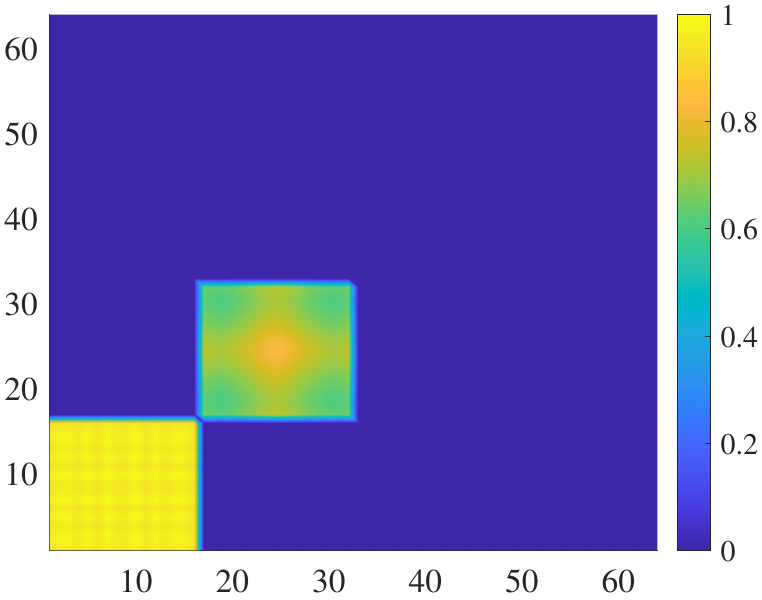}
			\label{fig_UE2_UAV2_Contraint12}
		\end{minipage}
	}
	\subfigure[$\mathbf{\tilde F}$: $N_{\text{ue}}=N_{\text{uav}}=2$, solution to $\mathcal{P}3'$.]{
		\begin{minipage}[t]{0.48\linewidth}
			\centering
			\includegraphics[width=3.3in]{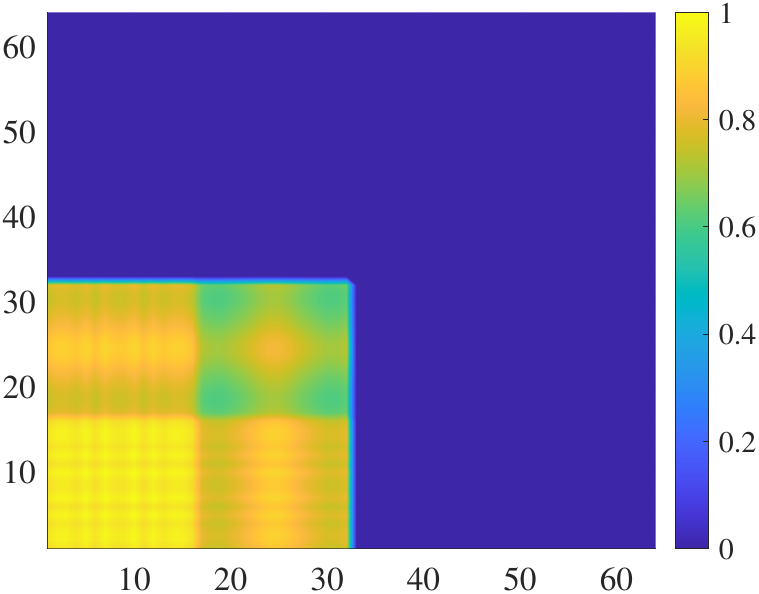}
			\label{fig_UE2_UAV2_Contraint123_SDP}
		\end{minipage}
	}
	\subfigure[$\mathbf{\bar F}$: $N_{\text{ue}}=N_{\text{uav}}=2$, output of Algorithm \ref{alg_JCJ}.]{
		\begin{minipage}[t]{0.48\linewidth}
			\centering
			\includegraphics[width=3.3in]{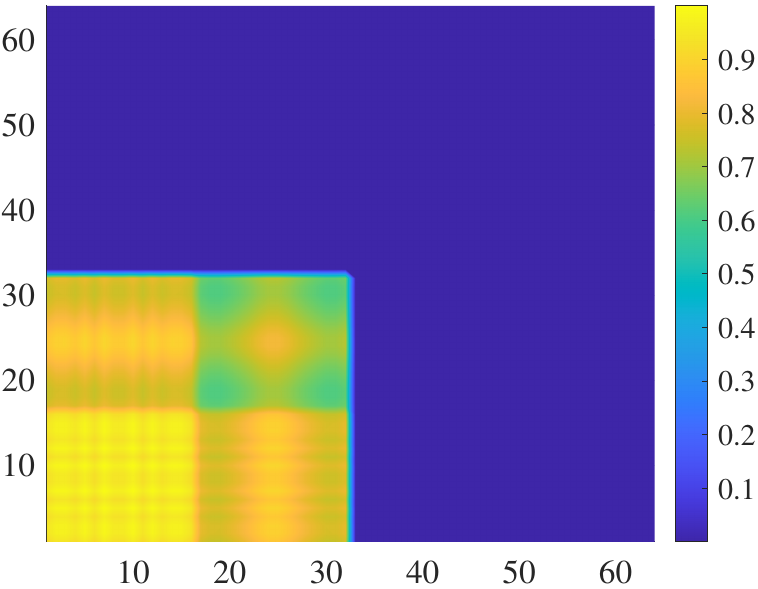}
			\label{fig_UE2_UAV2_Contraint123_SDP_EVD}
		\end{minipage}
	}
	\caption{The absolute value of $\mathbf{\mathord{\buildrel{\lower3pt\hbox{$\scriptscriptstyle\smile$}} \over F} }$, $\mathbf{\tilde F}$, and $\mathbf{\bar F}$ in a single realization, where $N_{\text{tx}}=16$. 
	(a) and (b) are the solution to $\mathcal{P}3$. (c) is the solution to $\mathcal{P}3'$. (d) is the output of Algorithm \ref{alg_JCJ}. (b)$\sim$(d) are investigated in the same system parameters settings, where $N_{\text{ue}}=N_{\text{uav}}=2$.
		In (a), $N_{\text{ue}}=0$ and $N_{\text{uav}}=4$.}
	\label{fig_rationality}
\end{figure*}
\begin{figure*}[!t]
	\vspace{-5mm}
	\centering
	\color{black}
	\subfigure[]{
		\begin{minipage}[t]{0.48\linewidth}
			\centering
			\includegraphics[width=3.3in]{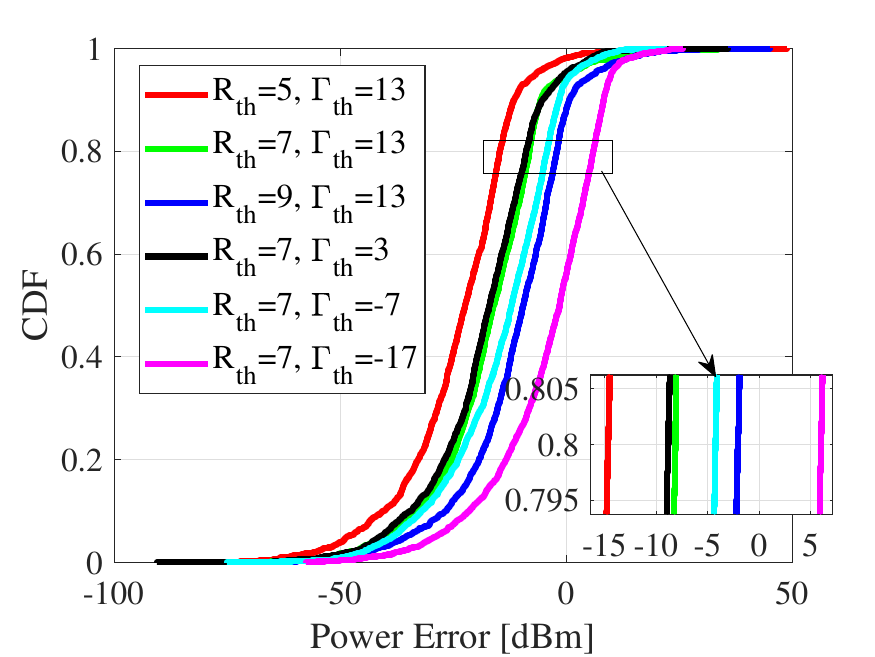}
			\label{fig_Threshold_PowerError}
		\end{minipage}
	}
	\subfigure[]{
		\begin{minipage}[t]{0.48\linewidth}
			\centering
			\includegraphics[width=3.3in]{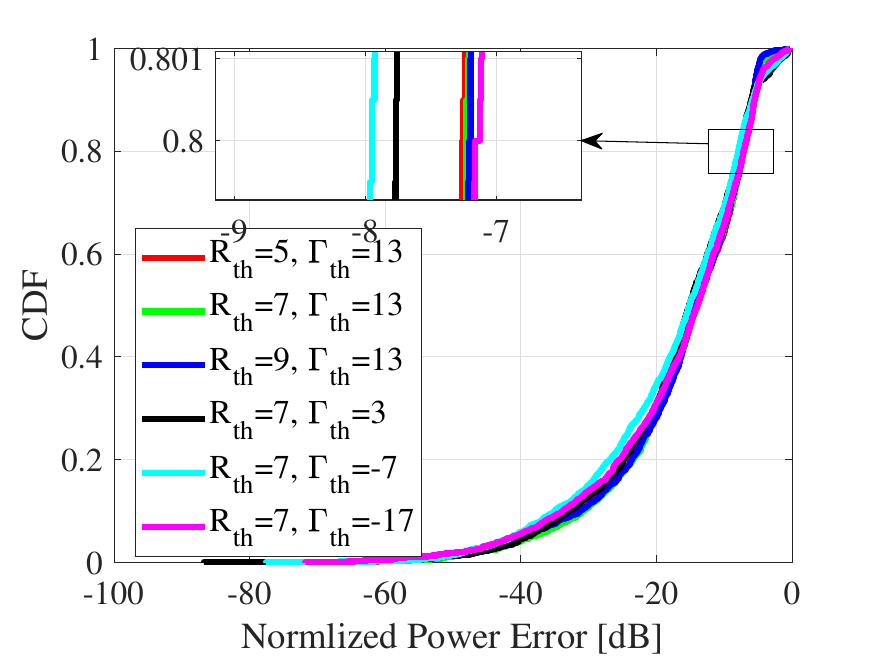}
			\label{fig_Threshold_PowerErrorNormalized}
		\end{minipage}
	}
	\subfigure[]{
		\begin{minipage}[t]{0.48\linewidth}
			\centering
			\includegraphics[width=3.3in]{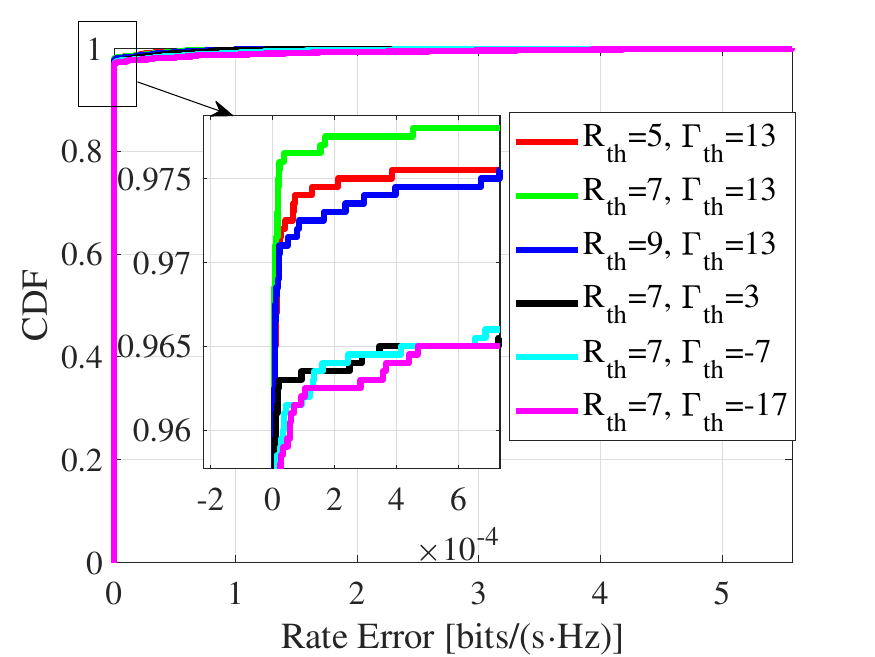}
			\label{fig_Threshold_RateError}
		\end{minipage}
	}
	\subfigure[]{
		\begin{minipage}[t]{0.48\linewidth}
			\centering
			\includegraphics[width=3.3in]{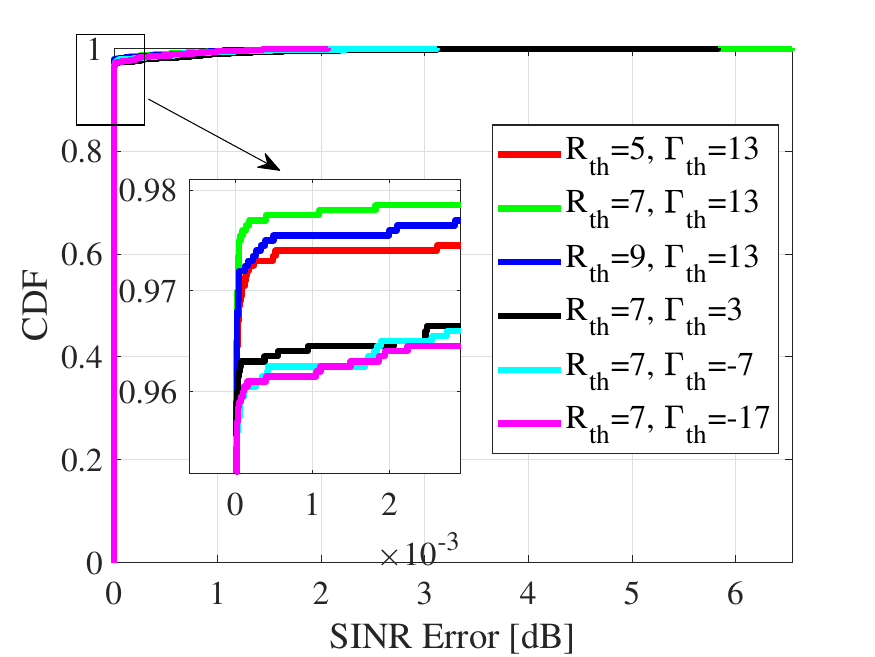}
			\label{fig_Threshold_SINRError}
		\end{minipage}
	}
	\caption{The performance of the JCJ scheme with varying $R_{\text{th}}$ and $\Gamma_{\text{th}}$.
		(a), (b), (c), and (d) are the performance of power error, normalized power error, rate error, and SINR error, respectively.}
	\label{fig_Threshold}
\end{figure*}
\begin{figure}[!t]
	%	\vspace{-5mm}
	\centering
	\color{black}
	\includegraphics[width=3.3in]{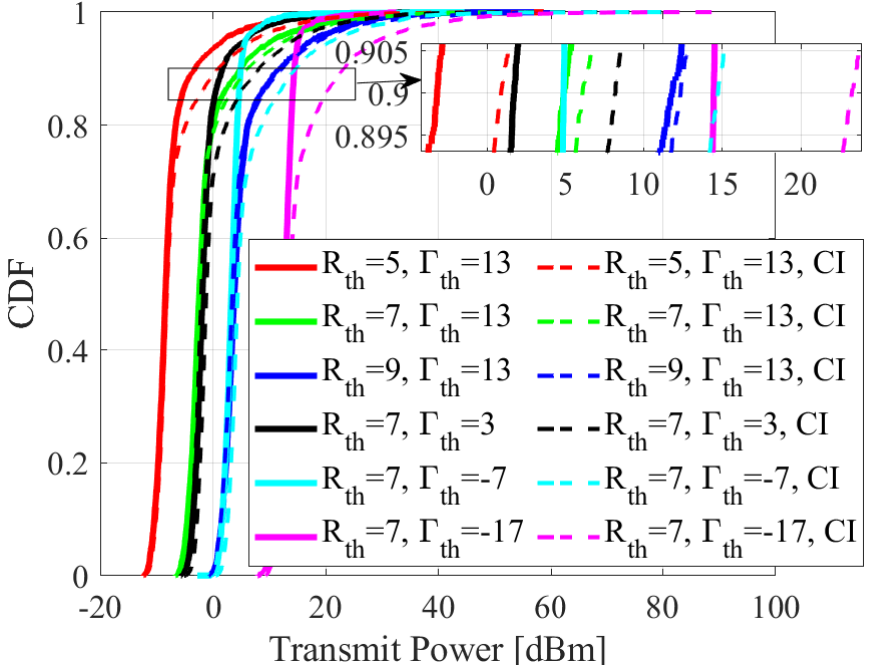}
	%	\vspace{-3mm}
	\caption{The transmit power difference of the JCJ scheme and the CI scheme with varying $R_{\text{th}}$ and $\Gamma_{\text{th}}$.
	}
	\label{fig_ChannelInversion_Threshold}
	%	\vspace{-6mm}
\end{figure}
\begin{figure*}[!t]
	\centering
	\color{black}
	\subfigure[]{
		\begin{minipage}[t]{0.48\linewidth}
			\centering
			\includegraphics[width=3.3in]{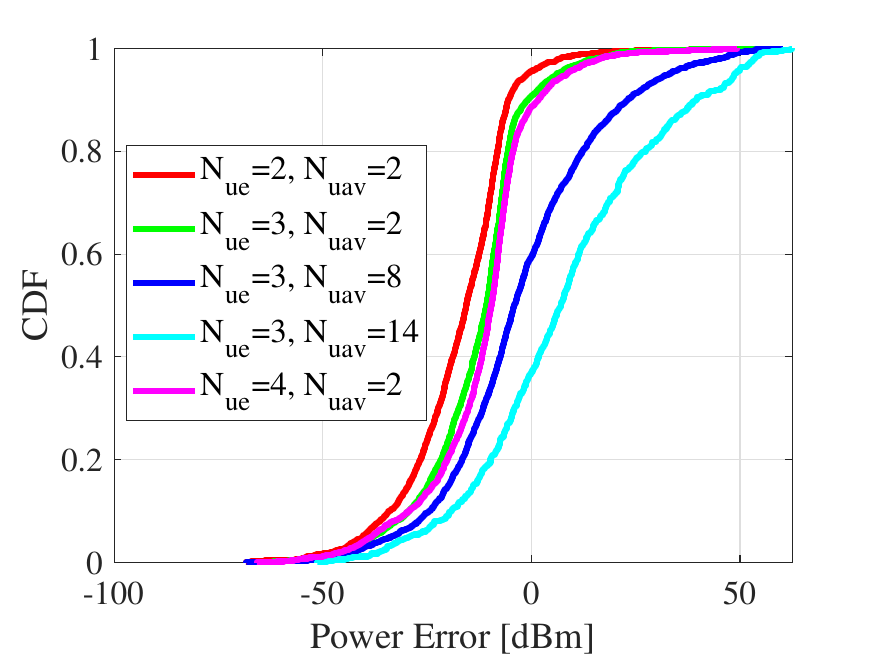}
			\label{fig_NumUEUAV_PowerError}
		\end{minipage}
	}
	\subfigure[]{
		\begin{minipage}[t]{0.48\linewidth}
			\centering
			\includegraphics[width=3.3in]{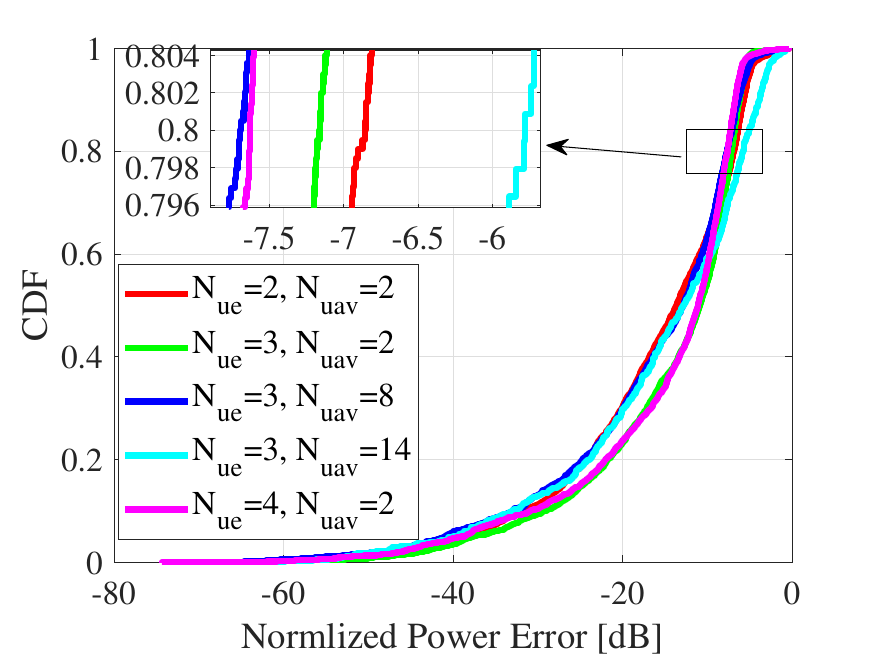}
			\label{fig_NumUEUAV_PowerErrorNormalized}
		\end{minipage}
	}
	\subfigure[]{
		\begin{minipage}[t]{0.48\linewidth}
			\centering
			\includegraphics[width=3.3in]{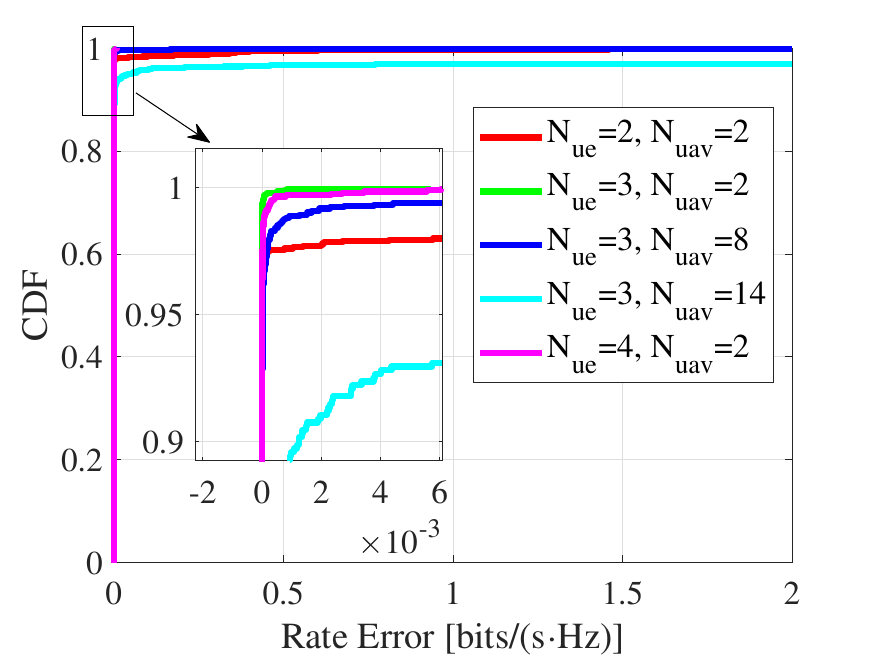}
			\label{fig_NumUEUAV_RateError}
		\end{minipage}
	}
	\subfigure[]{
		\begin{minipage}[t]{0.48\linewidth}
			\centering
			\includegraphics[width=3.3in]{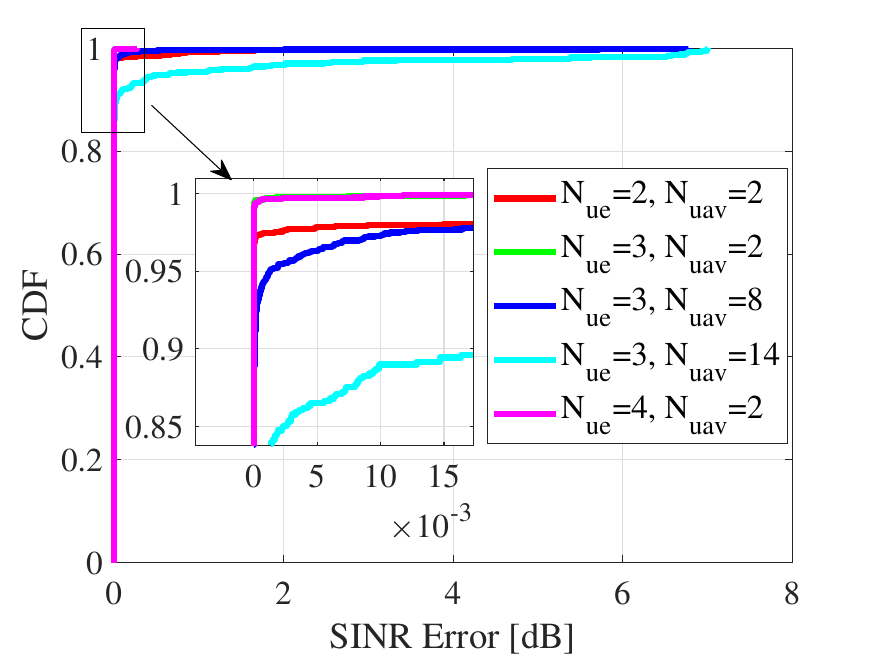}
			\label{fig_NumUEUAV_SINRError}
		\end{minipage}
	}
	\caption{The performance of the JCJ scheme with varying $N_{\text{ue}}$ and $N_{\text{uav}}$.
		(a), (b), (c), and (d) are the performance of power error, normalized power error, rate error, and SINR error, respectively.}
	\label{fig_NumUEUAV}
\end{figure*}
\begin{figure}[!t]
	%	\vspace{-5mm}
	\centering
	\color{black}
	\includegraphics[width=3.3in]{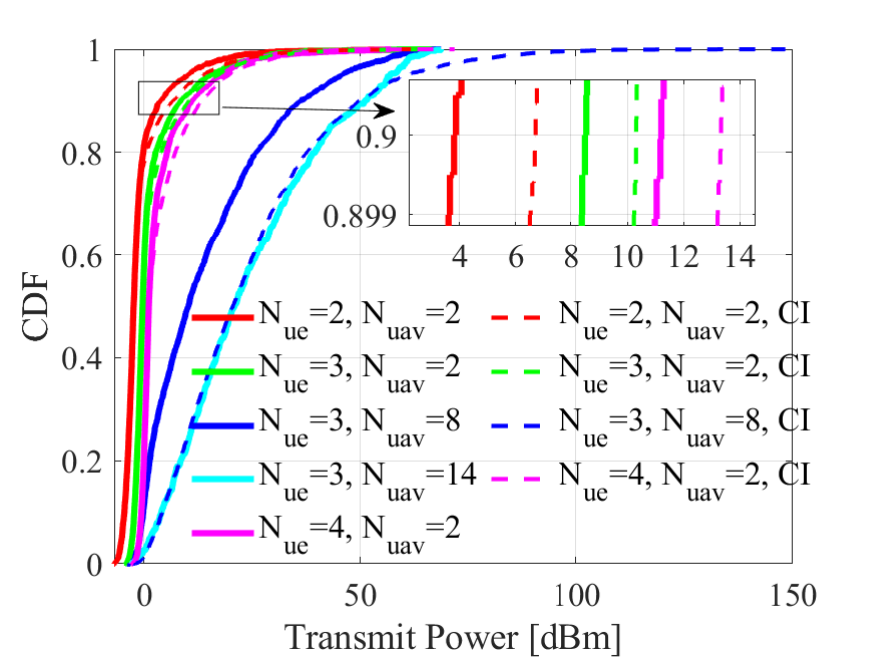}
	%	\vspace{-3mm}
	\caption{The transmit power difference of the JCJ scheme and the CI scheme with varying $N_{\text{ue}}$ and $N_{\text{uav}}$.}
	\label{fig_ChannelInversion_NumUEUAV}
	%	\vspace{-6mm}
\end{figure}
\begin{figure*}[!t]
	\centering
	\color{black}
	\subfigure[]{
		\begin{minipage}[t]{0.48\linewidth}
			\centering
			\includegraphics[width=3.3in]{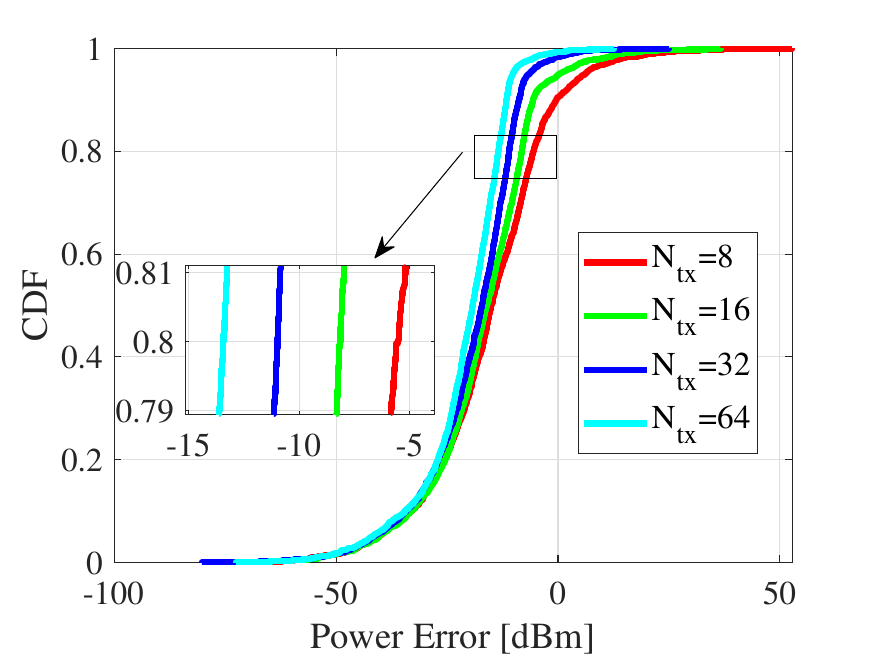}
			\label{fig_Nt_PowerError}
		\end{minipage}
	}
	\subfigure[]{
		\begin{minipage}[t]{0.48\linewidth}
			\centering
			\includegraphics[width=3.3in]{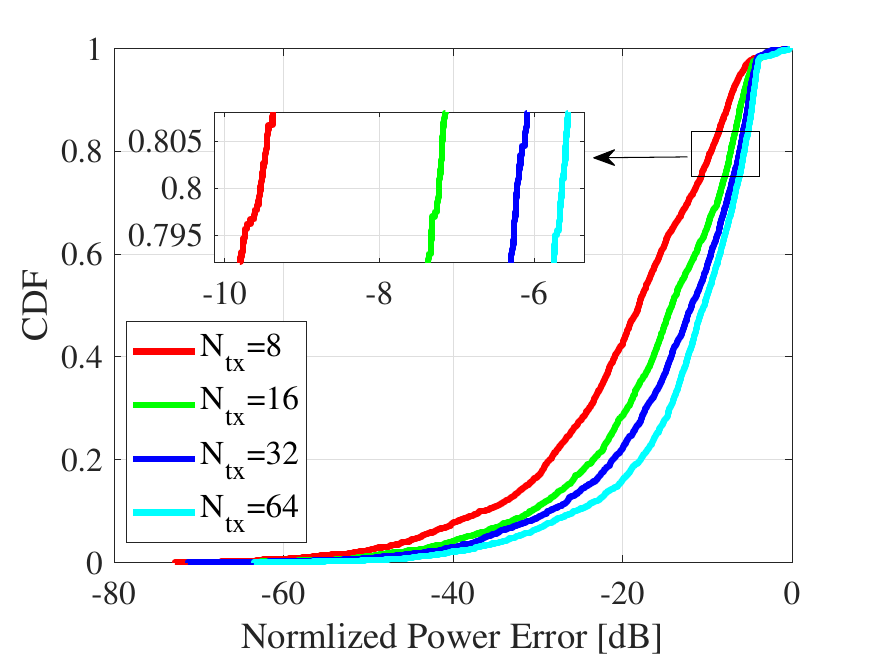}
			\label{fig_Nt_PowerErrorNormalized}
		\end{minipage}
	}
	\subfigure[]{
		\begin{minipage}[t]{0.48\linewidth}
			\centering
			\includegraphics[width=3.3in]{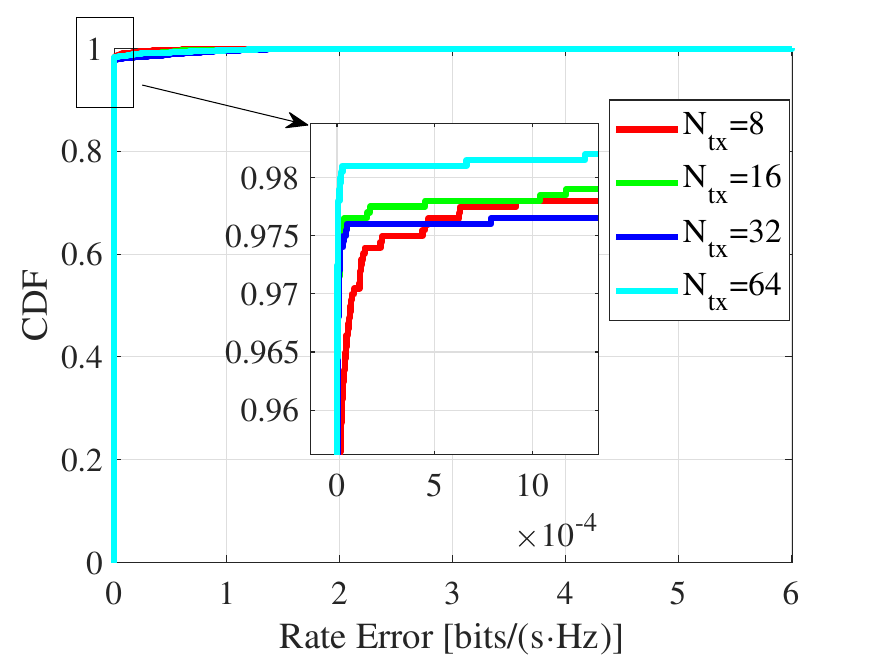}
			\label{fig_Nt_RateError}
		\end{minipage}
	}
	\subfigure[]{
		\begin{minipage}[t]{0.48\linewidth}
			\centering
			\includegraphics[width=3.3in]{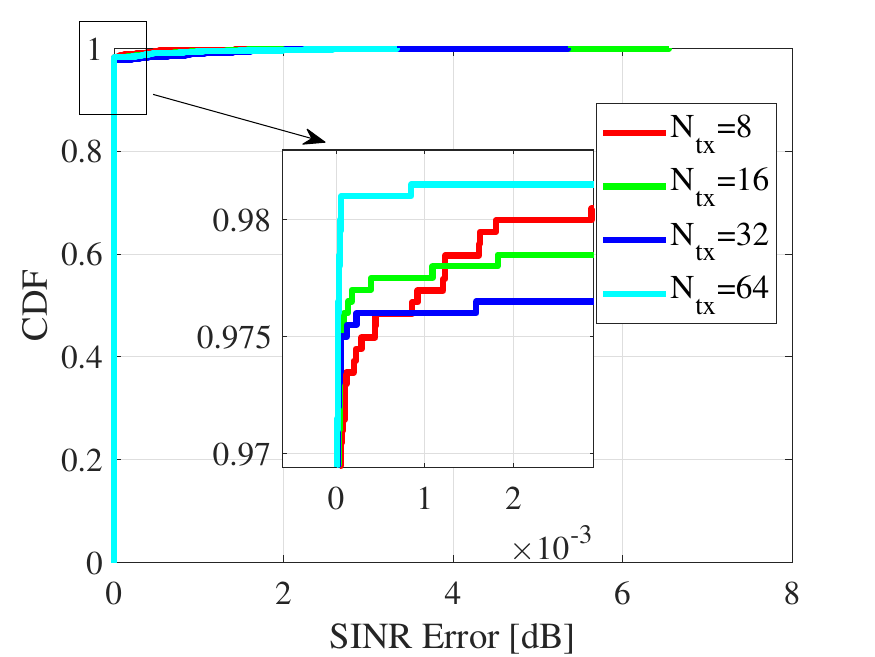}
			\label{fig_Nt_SINR}
		\end{minipage}
	}
	\caption{The performance of the JCJ scheme with varying $N_{\text{tx}}$.
		(a), (b), (c), and (d) are the performance of power error, normalized power error, rate error, and SINR error, respectively.}
	\label{fig_Nt}
\end{figure*}
\begin{figure}[!t]
	%	\vspace{-5mm}
	\centering
	\color{black}
	\includegraphics[width=3.3in]{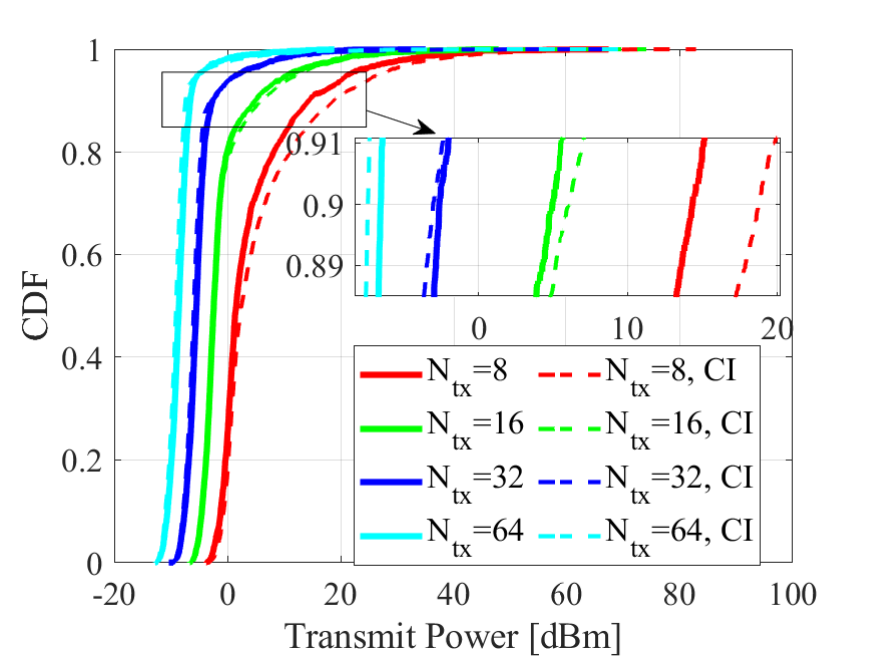}
	%	\vspace{-3mm}
	\caption{The transmit power difference of the JCJ scheme and the CI scheme with varying $N_{\text{tx}}$.}
	\label{fig_ChannelInversion_Nt}
	%	\vspace{-6mm}
\end{figure}
\subsection{System Parameters and Performance Metrics}
\begin{table}[!t]
	\caption{Simulation Parameters}	
	\label{table_para}
	\centering
	\begin{tabular}{llll}
		\hline \hline
		Meanning&Variable   & Value & Unit \\ \hline \hline
		Carrier Frequency&$f_c$       &    6   &         GHz                     \\
		Bandwidth&$B$          &   20    &       MHz         \\
		Number of UEs&$N_{\text{ue}}$      &   2    &         ---             \\
		{\begin{tabular}[l]{@{}l@{}}Number of \\ Unauthorized UAVs\end{tabular}}&$N_{\text{uav}}$       &   2    &         ---                \\
		{\begin{tabular}[l]{@{}l@{}}Expected Achievable \\ Rate Threshold\end{tabular}}  &$R_{\text{th},n}$       &   7    &        bit/(s$\cdot$Hz)              \\
		 {\begin{tabular}[l]{@{}l@{}}Expected SINR\\ Threshold\end{tabular}} &$\Gamma_{\text{th},m}$        &   13    &           dB             \\
		{\begin{tabular}[l]{@{}l@{}}Number of BS \\ Antennas\end{tabular}} &$N_\text{tx}$       &   16    &          ---             \\
		Range of UE (UAV)& $r_{\text{ue}}$ ($r_{\text{uav}}$)&     $\mathcal{U}(50,100)$ &             meter                \\
		AoD of UE (UAV)&$\theta_{\text{ue}}$ ($\theta_{\text{uav}}$) &  $\mathcal{U}(-60,60)$ &             degree               \\\hline \hline
	\end{tabular}
\end{table}
Key parameters are given in Table \ref{table_para}.
Unless otherwise specified, the parameter values listed in Table \ref{table_para} are used as the defaults.
$\boldsymbol\Phi=[0.01, 0.02, 0.03, 0.04, 0.05, 0.1, 0.2, 0.3, 0.4,$
$ 0.5, 0.6, 0.7, 0.8, 0.9]^T$.
The noise power spectrum density at the receiver is set as $N_0 =  - 174\;{\text{dBm/Hz}}$.
Therefore, $\sigma _{{\text{ue,}}n}^2=BN_0=-101\ \text{dBm},\;\forall n = 1,\ldots,{N_{{\text{ue}}}}$.
For simplicity, we also set $\sigma _{{\text{uav,}}m}^2=-101\  \text{dBm},\;\forall m = 1,\ldots,{N_{{\text{uav}}}}$.
$P_{\text{e},m}=-81$  dBm, $\forall m = 1,\ldots,{N_{{\text{uav}}}}$.
$\Gamma_{\text{th},m}$ is computed from the required BER of the $m$-th unauthorized UAV.
If the actual SINR of the $m$-th unauthorized UAV is less than or equal to $\Gamma_{\text{th},m}$, it can be assumed that $m$-th unauthorized UAV cannot communication reliably.
Specifically, the BER of $K$-quadrature amplitude modulation (QAM) can be computed as follows\cite{ref_DigitalComm}:
\begin{align}
	\text{BER} = 4Q\left(\sqrt{\frac{3\log_2K}{K-1}\frac{\mathcal{E}_{\text{bavg}}}{N_0}}\right),
\end{align}
where $K$ is the modulation order,
$\mathcal{E}_{\text{bavg}}$ is the energy per bit.
\begin{align}
	Q(x)=\frac12-\frac12\text{erf}\left(\frac x{\sqrt2}\right),
\end{align}
where $\text{erf}(\cdot)$ is the error function whose expression is 
\begin{align}
	\text{erf}(x)=\frac{2}{\sqrt{\pi}}\int_0^xe^{-t^2}dt.
\end{align}
We can obtain the relation between $\text{SNR}$ (or $\text{SINR}$) and $\text{BER}$ through the following transformation
\begin{align}
	\text{BER}&=4Q\left( {\sqrt {\frac{{3{{\log }_2}K}}{{K - 1}}\frac{{{{\cal E}_{{\text{bavg}}}}}}{{{N_0}}}} } \right)\nonumber\\
	&= 4Q\left( {\sqrt {\frac{3}{{K - 1}}\frac{{{{\cal E}_{{\text{avg}}}}}}{{{N_0}}}} } \right)\nonumber\\
	&= 4Q\left( {\sqrt {\frac{3}{{K - 1}}\frac{{{{\cal E}_{{\text{avg}}}}/(1/B)}}{{B{N_0}}}} } \right)\nonumber\\
	&= 4Q\left( {\sqrt {\frac{3}{{K - 1}}\frac{{{P_s}}}{{{P_n}}}} } \right) = 4Q\left( {\sqrt {\frac{3}{{K - 1}}{\text{SNR}}} } \right),	
\end{align}
where $\mathcal{E}_{\text{avg}}$ is the energy for each basic signal\cite{ref_DigitalComm}.
Consequently, to ensure unauthorized UAVs communication at a specified $\text{BER}$, $\text{SINR}$ must meet the following constraint
\begin{align}
	{\text{SINR}} \ge {\left[ {{Q^{ - 1}}\left( {\frac{{{\text{BER}}}}{4}} \right)} \right]^2}\frac{{K - 1}}{3}.	
\end{align}
If we set $\text{BER}=10^{-5}$ and $K=4$, the minimum $\text{SINR}$ is $13.19$ dB.
This accounts for the 13 dB listed in Table \ref{table_para}.

Since our goal is to minimize the transmit power, the transmit power serves as our performance metrics.
However, since locations of UEs and unauthorized UAVs are random, the transmit power varies in a large dynamic range.
Therefore, we adopt cumulative distributive function (CDF) to show the performance of the proposed JCJ scheme.
The data represented in the CDF curves from Fig. \ref{fig_Threshold} to Fig. \ref{fig_ChannelInversion_Nt} are obtained from 2000 random realizations.
For simplicity, we assume the uniform expected achievable rate threshold $R_{\text{th}}$ for all UEs and the uniform expected SINR threshold $\text{SINR}_{\text{th}}$ for all unauthorized UAVs.
We define the solution to $\mathcal{P}3$ as $\mathbf{\mathord{\buildrel{\lower3pt\hbox{$\scriptscriptstyle\smile$}} \over F} }$,
the solution to $\mathcal{P}3'$ as $\mathbf{\tilde F}$,
the output of Algorithm \ref{alg_JCJ} as $\mathbf{\bar F}$,
the actual achievable rate of the $n$-th UE as
\begin{align}
	{\bar {R}_n} = {\log _2}\left( {1 + \frac{{{{\bf{h}}_{{\text{ue,}}n}^H}{{\bf{\bar f}}_{{\text{ue,}}n}}{\bf{\bar f}}_{{\text{ue,}}n}^H{\bf{h}}_{{\text{ue,}}n}}}{{\sigma _{{\text{ue,}}n}^2 + {{\bf{h}}_{{\text{ue,}}n}^H}{\bf{\bar F}}{{\bf{\bar F}}^H}{\bf{h}}_{{\text{ue,}}n} - {{\bf{h}}_{{\text{ue,}}n}^H}{{\bf{\bar f}}_{{\text{ue,}}n}}{\bf{\bar f}}_{{\text{ue,}}n}^H{\bf{h}}_{{\text{ue,}}n}}}} \right),
\end{align}
where ${{\bf{\bar f}}_{{\text{ue,}}n}}={\bf{\bar F}}[:,n]$,
the actual SINR of the $m$-th unauthorized UAV as
\begin{align}
	{\bar {\Gamma} _{m}} = 10\log_{10} \frac{{{P_{{\text{e}},m}}}}{{{{\bf{h}}_{{\text{uav}},m}^H}{\bf{\bar F\bar F}}_{}^H{\bf{h}}_{{\text{uav}},m} + \sigma _{{\text{uav}},m}^2}}.
\end{align}
The power error in Fig. \ref{fig_Threshold_PowerError}, \ref{fig_NumUEUAV_PowerError}, and \ref{fig_Nt_PowerError} is defined as 
\begin{align}
	\text{Power Error} = |\text{tr}(\mathbf{\mathord{\buildrel{\lower3pt\hbox{$\scriptscriptstyle\smile$}} \over F} }\mathbf{\mathord{\buildrel{\lower3pt\hbox{$\scriptscriptstyle\smile$}} \over F} }^H) - \text{tr}(\mathbf{\bar F}\mathbf{\bar F}^H)|.
\end{align}
The normalized power error in Fig. \ref{fig_Threshold_PowerErrorNormalized}, \ref{fig_NumUEUAV_PowerErrorNormalized}, and \ref{fig_Nt_PowerErrorNormalized} is defined as 
\begin{align}
	\text{Normalized Power Error} = \frac{|\text{tr}(\mathbf{\mathord{\buildrel{\lower3pt\hbox{$\scriptscriptstyle\smile$}} \over F} }\mathbf{\mathord{\buildrel{\lower3pt\hbox{$\scriptscriptstyle\smile$}} \over F} }^H) - \text{tr}(\mathbf{\bar F}\mathbf{\bar F}^H)|}{\text{tr}(\mathbf{\bar F}\mathbf{\bar F}^H)}.
\end{align}
The rate error in Fig. \ref{fig_Threshold_RateError},  \ref{fig_NumUEUAV_RateError}, and \ref{fig_Nt_RateError} is defined as 
\begin{align}
	\mathop {\max }\limits_n \;|{\bar R_n} - {R_{{\text{th}}}}|.
\end{align}
The SINR error in Fig. \ref{fig_Threshold_SINRError}, \ref{fig_NumUEUAV_SINRError}, and \ref{fig_Nt_SINR} is defined as 
\begin{align}
	\mathop {\max }\limits_m \;|{\bar \Gamma _m} - {\Gamma _{{\text{th}}}}|.
\end{align}

\subsection{Baseline Scheme}

We adopt CI as our baseline scheme\cite{ref_03CM_MultiUserDnBF}.
We set $\mathbf{F}={{\bf{H}}^\dag }{\bf{\Lambda }}\in \mathbb{C}^{N_{\text{tx}}\times N_s}$, where ${{\bf{H}}^\dag }\in \mathbb{C}^{N_{\text{tx}}\times N_s}$ is the pseudo-inverse of $\mathbf{H}$, and $\boldsymbol{\Lambda }\in \mathbb{C}^{N_s\times N_s}$ is a diagonal matrix whose diagonal elements should be computed carefully to satisfy constraints $C_{1,1}$ and $C_{1,2}$.
Specifically, diagonal elements of $\bf{\Lambda }$ should satisfy the following constraints
\begin{align}
	{R_{{\text{th}}}} &= {\log _2}\left( {1 + \frac{{{{(\boldsymbol\Lambda [n,n])}^2}}}{{\sigma _{{\text{ue}},n}^2}}} \right),\;n = 1,\ldots,{N_{{\text{ue}}}},\\
	{\Gamma _{{\text{th}}}} &= 10{\log _{10}}\left( {\frac{{{P_{{\text{e}},m}}}}{{{{(\boldsymbol\Lambda [{N_{{\text{ue}}}} + m,{N_{{\text{ue}}}} + m])}^2} + \sigma _{{\text{uav}},m}^2}}} \right),\nonumber\\
	&\qquad\qquad\qquad	\qquad\qquad\qquad m = 1,\ldots,{N_{{\text{uav}}}}.
\end{align}

\subsection{Rationality of Problem Reformulation}\label{sec_rationality}
In Fig. \ref{fig_UE0_UAV4_Contraint12}, we can see that if there is no legitimate UE, the required number of streams should be $N_{\text{uav}}$.
Comparing Fig. \ref{fig_UE0_UAV4_Contraint12} with Fig. \ref{fig_UE2_UAV2_Contraint12}, \ref{fig_UE2_UAV2_Contraint123_SDP}, and \ref{fig_UE2_UAV2_Contraint123_SDP_EVD}, Corollary \ref{corollary_2} can be verified: if $N_{\text{ue}}\ge 1$, additional jamming streams are not required.

In Fig. \ref{fig_UE2_UAV2_Contraint12}, we can see that ${\bf{\tilde F}} = {\text{blkdiag}}({{\bf{S}}_1},\ldots,{{\bf{S}}_{{N_{{\text{ue}}}}}},$
${\bf{0}}_{N_{\text{uav}}N_{\text{tx}}\times N_{\text{uav}}N_{\text{tx}}})$,
where $\mathbf{S}_{n},n=1,\ldots,N_\text{ue}$ is an $N_\text{tx}\times N_\text{tx}$ symmetric matrix.
Therefore, Theorem \ref{theorem_1} is verified.

Comparing Fig. \ref{fig_UE2_UAV2_Contraint123_SDP} and \ref{fig_UE2_UAV2_Contraint123_SDP_EVD}, we can qualitatively analyze and ascertain that the solution to $\mathcal{P}3'$ and the output of Algorithm \ref{alg_JCJ} have no significant differences.
Quantitative analyses will be provided in Sections \ref{sec_sim_threshold},  \ref{sec_sim_NumUEUAV}, and \ref{sec_sim_Nt}.

\subsection{Impact of $R_{\text{th}}$ and $\Gamma_{\text{th}}$ on Transmit Power}\label{sec_sim_threshold}
In Fig. \ref{fig_Threshold_PowerError}, we can see that the larger $R_{\text{th}}$ or the smaller $\Gamma_{\text{th}}$, the larger power error.
Meanwhile, we can see the power error range is wide.
This can be attributed to the introduction of the additional constraint $C_{3,3}$ in problem $\mathcal{P}3'$.
While this constraint significantly increases the likelihood of satisfying the constraints of problem $\mathcal{P}2$, it also raises the transmit power.
However, considering only the power error is not comprehensive, as the transmit power itself is inherently higher when UEs' demands for the achievable rate or the SINR required to interfere with unauthorized UAVs is higher.
Therefore, the performance of normalized power error is indispensable.

In Fig. \ref{fig_Threshold_PowerErrorNormalized}, when $R_{\text{th}}$ and $\Gamma_{\text{th}}$ change, the variation of normalized power error is significantly less than that of the power error.
At a CDF value of 0.8, the normalized power error ranges from a minimum of $-8 \text{ dB}$ to a maximum of $-7 \text{ dB}$ for different values of $R_{\text{th}}$ and $\Gamma_{\text{th}}$.
This indicates that the proposed JCJ scheme serves as a rational approximation to $\mathcal{P}2$.
On the one hand, the solution of the proposed JCJ scheme satisfies the rank-1 constraint.
On the other hand, the power requirement for $\mathcal{P}3$ is less than or equal to that of $\mathcal{P}2$, since $\mathcal{P}2$ has an additional rank-1 constraint compared to $\mathcal{P}3$.

In Figs. \ref{fig_Threshold_RateError} and \ref{fig_Threshold_SINRError}, when the requirements for $R_{\text{th}}$ and $\Gamma_{\text{th}}$ change, the actual achievable rates and SINRs differ very slightly from the expected achievable rate threshold and SINR threshold.
This indicates that the proposed JCJ scheme is effective and applying eigenvalue decomposition to the solution to $\mathcal{P}3'$ as an approximation to the solution to $\mathcal{P}2$ is feasible.
Although the actual achievable rate and SINR deviate very slightly from the expected achievable rate and SINR, when zoomed in, it can be observed that an increased number of unauthorized UAVs leads to slightly larger deviations.

In Fig. \ref{fig_ChannelInversion_Threshold}, we can see that the proposed JCJ scheme has 1$\sim$8 dB gain over the conventional CI scheme in transmit power,
demonstrating the integration gain of the proposed JCJ scheme.
Integration here means the joint design of the beamformer for UEs and the beamformer for unauthorized UAVs.
When $R_{\text{th}}$ or $\Gamma_{\text{th}}$ decreases, the gain of the JCJ scheme compared with the CI scheme increases.
The inferiority of the CI scheme to the proposed JCJ scheme comes from the complete independence among different streams, and the need for additional jamming streams to counteract unauthorized UAVs.
Unlike the CI scheme, our proposed JCJ scheme does not require complete independence among streams.
Introducing appropriate levels of jamming while still meeting the overall achievable rate and SINR constraints is entirely feasible.

\subsection{Impact of $N_{\text{ue}}$ and $N_{\text{uav}}$ on Transmit Power}\label{sec_sim_NumUEUAV}
In Fig. \ref{fig_NumUEUAV_PowerError}, the more UEs and unauthorized UAVs, the larger the power error.
In Fig. \ref{fig_NumUEUAV_PowerErrorNormalized}, the normalized power error remains roughly invariant as the number of UEs and unauthorized UAVs are varied.
The normalized power error ranges from a minimum of $-7.5 \text{ dB}$ to a maximum of $-6 \text{ dB}$ for different numbers of UEs and unauthorized UAVs at a CDF value of 0.8.
This indicates that the proposed JCJ scheme serves as a rational approximation to $\mathcal{P}2$.

In Fig. \ref{fig_NumUEUAV_RateError}, when $N_{\text{ue}}$ is varied, the difference between the actual achievable rate and the expected achievable rate is small.
As $N_{\text{uav}}$ increases, the difference between the actual achievable rate and the expected achievable rate will increase.
However, there is still a 0.9 probability that the achievable rate error is within 0.001 bits/(s$\cdot$Hz), indicating that applying eigenvalue decomposition to $\mathcal{P}3'$ as an approximation to the solution of $\mathcal{P}2$ is feasible.

In Fig. \ref{fig_NumUEUAV_SINRError}, the variation of $N_{\text{ue}}$ has a smaller impact on the difference between the actual SINR and the expected SINR compared to the variation of $N_{\text{uav}}$.
However, even when $N_{\text{uav}}$ is equal to 14, the largest SINR error still has a 0.9 probability of being within 0.017 dB, which suggests that applying eigenvalue decomposition to $\mathcal{P}3'$ in order to approximate the solution to $\mathcal{P}2$ is feasible.

In Fig. \ref{fig_ChannelInversion_NumUEUAV}, when $N_{\text{ue}}$ changes, we can see that the proposed JCJ scheme has 2$\sim$3 dB gain over the conventional CI scheme in transmit power,
demonstrating the integration gain of the proposed JCJ scheme.
With the increase of $N_{\text{uav}}$ from 2 to 8, the performance gap expands from 2 dB to 13 dB, showing that the proposed JCJ scheme becomes increasingly advantageous as the number of unauthorized UAVs grows.
Notably, when $N_{\text{ue}}$ is set to 3 and $N_{\text{uav}}$ to 14, such that the total number of devices $N_{\text{ue}} + N_{\text{uav}} = 17$ exceeds the number of transmit antennas $N_{\text{tx}} = 16$, the CI scheme fails to work.

\subsection{Impact of Number of Transmit Antennas on Transmit Power}\label{sec_sim_Nt}
In Fig. \ref{fig_Nt_PowerError}, the larger $N_{\text{tx}}$, the smaller power error.
This is because the more antennas there are, the narrower the beam becomes, allowing for anticipated energy in the desired direction while minimizing energy leakage in other directions.
In Fig. \ref{fig_Nt_PowerErrorNormalized}, the larger $N_{\text{tx}}$, the larger normalized power error.
This is because, while the transmit power error decreases, the overall transmit power also reduces, as illustrated in Fig. \ref{fig_Nt}.
With $N_{\text{tx}}=8$, there is a 0.8 probability that the normalized power error is below $-9.5$ dB, indicating that the proposed JCJ scheme becomes more effective in small antenna systems.

In Fig. \ref{fig_Nt_RateError} and \ref{fig_Nt_SINR}, with different values of $N_{\text{tx}}$, the differences between the actual achievable rates and the expected achievable rates are minimal.
The same holds true for the SINR error.
Furthermore, these differences continue to decrease as $N_{\text{tx}}$ increases, which is attributed to the narrowing of the beam's lobe.

From Fig. \ref{fig_ChannelInversion_Nt}, we can observe that as the number of antennas increases, the gain of the JCJ scheme gradually diminishes.
When the number of antennas reaches 32, the CI scheme surpasses the JCJ scheme, indicating that with a larger number of antennas, the improved angular resolution facilitates lower transmit power for independent streams.
Therefore, the proposed JCJ scheme is advantageous in scenarios with a small number of antennas or when the combined number of UEs and unauthorized UAVs exceeds the number of antennas.
\begin{Remark}
	From Fig. \ref{fig_ChannelInversion_Nt}, it appears counter-intuitive that a higher number of antennas is associated with lower power consumption, as an increased number of antennas implies more radio frequency chains, which are known to be power-intensive.
	This is due to the fact that our simulation does not take into account the static power consumption of the radio frequency chains, as this factor does not affect the performance comparison between the proposed JCJ scheme and the CI scheme.
\end{Remark}

\begin{Remark}
	In our simulations, we ensure that the angular separation between UEs is at least 5° to avoid making $\mathcal{P}1$ unsolvable.
	Additionally, we impose equality constraints in the first two inequality constraints when solving $\mathcal{P}3'$ to constrain the solution's degrees of freedom, thereby ensuring that the eigenvalue decomposition of the solution to $\mathcal{P}3'$ satisfies the constraints of $\mathcal{P}2$.
	Moreover, some instances of excessively large transmit power observed in the simulation can be attributed to insufficient angular separation between UEs and unauthorized UAVs.
\end{Remark}

\subsection{Insights}\label{sec_insight}
\begin{itemize}
	\item Fig. \ref{fig_rationality}: Dedicated jamming streams for countering unauthorized UAVs are not necessary.
	
	\item Figs. \ref{fig_ChannelInversion_Threshold}, \ref{fig_ChannelInversion_NumUEUAV}, and \ref{fig_ChannelInversion_Nt}: The lower the expected SINR threshold, the higher the number of unauthorized UAVs, and the fewer the transmit antennas, the more pronounced is the advantage of the proposed JCJ scheme over the CI scheme.
	The most significant advantage of the proposed JCJ scheme is its ability to operate effectively when the total number of UEs and unauthorized UAVs exceeds the number of transmit antennas.
\end{itemize}

\section{Conclusion}\label{sec_Conclusion}
To leverage the widely deployed MIMO BSs for ubiquitous countermeasures against unauthorized UAVs, this paper focused on the joint design of beamforming to establish a dual-functional system, simultaneously communicating with legitimate UEs while jamming unauthorized UAVs.
First, we have formulated the JCJ problem and then applied SDR to convert it into a standard SDP problem.
By analyzing the structure of the formulated problem, we have found dedicated jamming streams for countering unauthorized UAVs are not necessary.
Then, we have proposed to add an innovative constraint that ensures the solution of the reformulated problem satisfy the rank-1 constraint.
Finally, we have verified the correctness and effectiveness of the proposed JCJ scheme through extensive simulations, from which some interesting insights are obtained.
The results confirm that the proposed JCJ scheme is advantageous in scenarios with a small number of antennas or when the combined number of UEs and unauthorized UAVs exceeds the number of antennas.

However, the JCJ scheme proposed in this paper did not incorporate sensing functionalities.
The joint design of communication, sensing, and countermeasures presents a promising avenue for future research.
Furthermore, investigating how to coordinate communication, sensing, and jamming when the communication and jamming frequency bands differ in more practical scenarios is another worthwhile direction for future exploration.
Additionally, as for the array structure of the BS using hybrid beamforming, the proposed approach could serve as an input for many existing hybrid beamforming solutions, it may increase complexity and potentially lead to suboptimal performance. 
Therefore, directly designing a hybrid beamformer to jam unauthorized UAVs using communication signals, based on the problem’s structure, presents a valuable direction for future research.
Moreover, leveraging the beamfocusing property of the extremely-large array may lead more effective jamming task.

\begin{appendices}
	\setcounter{equation}{0}
	\renewcommand{\theequation}{A-\arabic{equation}}
	\section{Proof of Theorem 1}\label{appendix_1}
	Prior to presenting the proof of Theorem \ref{theorem_1}, we first provide Lemmas \ref{lemma_1} and \ref{lemma_2}.
	\begin{Lemma}\label{lemma_1}
		For any positive semi-definite matrix $\mathbf{A}\in\mathbb{C}^{K\times K}$, ${\bf{A}}[{i_1}:{i_2},{i_1}:{i_2}] \in {\mathbb{C}^{({i_2} - {i_1} + 1) \times ({i_2} - {i_1} + 1)}}$	is also a positive semi-definite matrix, where $i_2\ge i_1, i_1\in[1,K], i_2\in[1,K]$.
	\end{Lemma}
	\begin{proof}
		The eigenvalue decomposition of $\mathbf{A}$ is given as 
		\begin{align}
			\mathbf{A}=\sum_{k=1}^{K}\alpha_k \mathbf{a}_k\mathbf{a}_k^H,
		\end{align}
		where $\alpha_k \ge0$.
		Then, $\mathbf{\tilde A}={\bf{A}}[{i_1}:{i_2},{i_1}:{i_2}]
		=\sum_{i=k}^{K}\alpha_k \mathbf{a}_k[{i_1}:{i_2}]\mathbf{a}_k^H[{i_1}:{i_2}]$.
		For $\forall\ \mathbf{b}\in\mathbb{C}^{i_2-i_1+1}$, 
		\begin{align}
				\mathbf{b}^H \mathbf{\tilde A}\mathbf{b}&=
				\sum_{k=1}^{K}\alpha_k \mathbf{b}^H\mathbf{a}_k[{i_1}:{i_2}]\mathbf{a}_k^H[{i_1}:{i_2}]\mathbf{b}\nonumber\\
				& =\sum_{k=1}^{K}\alpha_k |q_k|^2 \ge 0,
		\end{align}
		where $q_k = \mathbf{a}_k^H[{i_1}:{i_2}]\mathbf{b}$.
		Therefore, ${\bf{A}}[{i_1}:{i_2},{i_1}:{i_2}]$ is also a positive semi-definite matrix.
		\qedhere
	\end{proof}
	\begin{Lemma}\label{lemma_2}
		For any positive semi-definite matrix $\mathbf{A}_1\in\mathbb{C}^{K\times K}$, negative semi-definite matrix $\mathbf{A}_2\in\mathbb{C}^{K\times K}$, $\text{tr}(\mathbf{A}_1\mathbf{A}_2)\le 0$.
	\end{Lemma}
	\begin{proof}
		The eigenvalue decomposition of $\mathbf{A}_1$ and $\mathbf{A}_2$ are given as 
		\begin{align}
			\mathbf{A}_1&=\sum\nolimits_{k=1}^{K}\alpha_{1,k} \mathbf{a}_{1,k}\mathbf{a}_{1,k}^H,\\
			\mathbf{A}_2&=\sum\nolimits_{k'=1}^{K}\alpha_{2,k'} \mathbf{a}_{2,k'}\mathbf{a}_{2,k'}^H,
		\end{align}
		where $\alpha_{1,k} \ge0$ and $\alpha_{2,k'} \le0$.
		\begin{align}
		{\text{tr(}}{{\bf{A}}_1}{{\bf{A}}_2}) &= {\text{tr}}\left( {\sum\limits_{k = 1}^K {{\alpha _{1,k}}} {{\bf{a}}_{1,k}}{\bf{a}}_{1,k}^H\sum\limits_{k' = 1}^K {{\alpha _{2,k'}}} {{\bf{a}}_{2,k'}}{\bf{a}}_{2,k'}^H} \right)\nonumber\\
		&= {\text{tr}}\left( {\sum\limits_{k = 1}^K {\sum\limits_{k' = 1}^K {{\alpha _{1,k}}{\alpha _{2,k'}}{{\bf{a}}_{1,k}}{\bf{a}}_{1,k}^H{{\bf{a}}_{2,k'}}{\bf{a}}_{2,k'}^H} } } \right)\\
		&= \sum\limits_{k = 1}^K {\sum\limits_{k' = 1}^K {{\alpha _{1,k}}{\alpha _{2,k'}}} } {\text{tr}}\left( {{{\bf{a}}_{1,k}}{\bf{a}}_{1,k}^H{{\bf{a}}_{2,k'}}{\bf{a}}_{2,k'}^H} \right)\nonumber\\
		&= \sum\limits_{k = 1}^K {\sum\limits_{k' = 1}^K {{\alpha _{1,k}}{\alpha _{2,k'}}} } {\text{tr}}\left( {{\bf{a}}_{1,k}^H{{\bf{a}}_{2,k'}}{\bf{a}}_{2,k'}^H{{\bf{a}}_{1,k}}} \right)\nonumber\\
		 &= \sum\limits_{k = 1}^K {\sum\limits_{k' = 1}^K {{\alpha _{1,k}}{\alpha _{2,k'}}} } |{\bf{a}}_{1,k}^H{{\bf{a}}_{2,k'}}{|^2} \le 0\nonumber.
		\end{align}
		\qedhere
	\end{proof}
	
	Since ${{\bf{\tilde H}}_{1,n}^{}}$, ${{\bf{\bar H}}_m^{}}$, and ${\bf{\tilde H}}_{2,n}^H{\bf{\tilde H}}_{2,n}^{}$ are block-diagonal, $\mathbf{A}_{1,n}$ and $\mathbf{A}_{2,m}$ are also block-diagonal.
	Therefore, the values of the off-block-diagonal elements of $\mathbf{\tilde F}$ do not affect the solution to $\mathcal{P}3$, and $\mathbf{\tilde F}$ is block-diagonal if its initial values of the off-block-diagonal elements are zero.
	
	Since the sub-blocks of ${{\bf{A}}_{1,n}}[{N_{{\text{ue}}}}{N_{{\text{tx}}}} + 1:{N_s}{N_{{\text{tx}}}}{\text{,}}\;{N_{{\text{ue}}}}{N_{{\text{tx}}}} + 1:{N_s}{N_{{\text{tx}}}}]$ are negative definite matrices, i.e.,
	\begin{align}
		\begin{array}{l}
			{{\bf{A}}_{1,n}}[{N_{{\text{ue}}}}{N_{{\text{tx}}}} + 1 + (i - 1){N_{{\text{tx}}}}:{N_{{\text{ue}}}}{N_{{\text{tx}}}} + i{N_{{\text{tx}}}}{\text{,}}\\
			\;{N_{{\text{ue}}}}{N_{{\text{tx}}}} + 1 + (i - 1){N_{{\text{tx}}}}:{N_{{\text{ue}}}}{N_{{\text{tx}}}} + i{N_{{\text{tx}}}}]\\ \qquad\qquad\qquad \forall i=1,\ldots,N_{\text{uav}},
		\end{array}
	\end{align}
	are negative definite matrices, we can obtain the conclusion as
	\begin{align}\label{eq_Appendix_1}
		\begin{array}{l}
			{{\bf{\tilde F}}}[{N_{{\text{ue}}}}{N_{{\text{tx}}}} + 1 + (i - 1){N_{{\text{tx}}}}:{N_{{\text{ue}}}}{N_{{\text{tx}}}} + i{N_{{\text{tx}}}}{\text{,}}\\
			\;{N_{{\text{ue}}}}{N_{{\text{tx}}}} + 1 + (i - 1){N_{{\text{tx}}}}:{N_{{\text{ue}}}}{N_{{\text{tx}}}} + i{N_{{\text{tx}}}}]=\mathbf{0}_{N_{\text{tx}}\times N_{\text{tx}}}\\ \qquad\qquad\qquad \forall i=1,\ldots,N_{\text{uav}}.
		\end{array}
	\end{align}
	We can employ the method of proof by contradiction to prove (\ref{eq_Appendix_1}).
	If ${\bf{\tilde F}}[{N_{{\text{ue}}}}{N_{{\text{tx}}}} + 1:{N_s}{N_{{\text{tx}}}}{\text{,}}\;{N_{{\text{ue}}}}{N_{{\text{tx}}}} + 1:{N_s}{N_{{\text{tx}}}}] \ne {\bf{0}}_{N_{\text{uav}}N_{\text{tx}}\times N_{\text{uav}}N_{\text{tx}}}$, we assume that there exists $\tilde i$ that 
	\begin{align}
		\begin{array}{l}
			{\bf{\tilde F}}'={{\bf{\tilde F}}}[{N_{{\text{ue}}}}{N_{{\text{tx}}}} + 1 + (\tilde i - 1){N_{{\text{tx}}}}:{N_{{\text{ue}}}}{N_{{\text{tx}}}} + \tilde i{N_{{\text{tx}}}}{\text{,}}\\
			\;{N_{{\text{ue}}}}{N_{{\text{tx}}}} + 1 + (\tilde i - 1){N_{{\text{tx}}}}:{N_{{\text{ue}}}}{N_{{\text{tx}}}} + \tilde i{N_{{\text{tx}}}}]\ne\mathbf{0}_{N_{\text{tx}}\times N_{\text{tx}}}.
		\end{array}
	\end{align}
	Then, we define
	\begin{align}
		\begin{array}{l}
			{\bf{A}}'={{\bf{A}}_{1,n}}[{N_{{\text{ue}}}}{N_{{\text{tx}}}} + 1 + (\tilde i - 1){N_{{\text{tx}}}}:{N_{{\text{ue}}}}{N_{{\text{tx}}}} + \tilde i{N_{{\text{tx}}}}{\text{,}}\\
			\;{N_{{\text{ue}}}}{N_{{\text{tx}}}} + 1 + (\tilde i - 1){N_{{\text{tx}}}}:{N_{{\text{ue}}}}{N_{{\text{tx}}}} + \tilde i{N_{{\text{tx}}}}].
		\end{array}
	\end{align}
	Through Lemma \ref{lemma_1}, we know that ${\bf{\tilde F}}'$ is positive semi-definite matrix.
	Through Lemma \ref{lemma_2}, we know that $\text{tr}({\bf{ A}}'{\bf{\tilde F}}')\le 0$.
	Therefore, we need the eigenvalue of
	\begin{align}
		\begin{array}{l}
			{\bf{\tilde F}}[1 + (i - 1){N_{{\text{tx}}}}:i{N_{{\text{tx}}}}{\text{,}}\;1 + (i - 1){N_{{\text{tx}}}}:i{N_{{\text{tx}}}}],
			\forall i = 1,\ldots,{N_{{\text{ue}}}},\nonumber
		\end{array}
	\end{align}
	to be larger to satisfy the constraint $C_{3,1}$ of $\mathcal{P}3$.
	However, this is contrary to the objective of $\mathcal{P}3$.
	Therefore, (\ref{eq_Appendix_1}) holds true, and Theorem \ref{theorem_1} is proved.
	
	\section{Proof of Theorem 2}\label{appendix_2}
	Notice the constraint $\text{rank}(\mathbf{\tilde F})=1$, $\mathbf{\tilde F}=\mathbf{\tilde f}\mathbf{\tilde f}^H$, where $\mathbf{\tilde f}\in \mathbb{C}^{N_sN_{\text{tx}}}$.
	Let the maximum index of non-zero value of $\mathbf{\tilde f}$ be $i_{\text{max}}$.
	Then, the proof of Theorem \ref{theorem_2} can be transformed to prove $i_{\text{max}}\le N_{\text{ue}}N_{\text{tx}}$.
	We can employ the method of proof by contradiction to prove $i_{\text{max}}\le N_{\text{ue}}N_{\text{tx}}$.
	
	Let $i_{\text{max}}> N_{\text{ue}}N_{\text{tx}}$ and $i_{\text{max}} = N_{\text{ue}}N_{\text{tx}} + \tilde i$, where $\tilde i \ge 1$.
	By defining
	\begin{align}
		{\bf{\tilde F}}' &= {\bf{\tilde F}}[{N_{{\text{ue}}}}{N_{{\text{tx}}}} + \tilde i:{N_s}{N_{{\text{tx}}}}{\text{,}}\;{N_{{\text{ue}}}}{N_{{\text{tx}}}} + \tilde i:{N_s}{N_{{\text{tx}}}}] \ne {\bf{0}}\nonumber,\\
		{\bf{A}}' &= {{\bf{A}}_{1,n}}[{N_{{\text{ue}}}}{N_{{\text{tx}}}} + \tilde i:{N_s}{N_{{\text{tx}}}}{\text{,}}\;{N_{{\text{ue}}}}{N_{{\text{tx}}}} + \tilde i:{N_s}{N_{{\text{tx}}}}],
	\end{align}
	we know that ${\bf{\tilde F}}'$ is a positive semi-definite matrix through Lemma \ref{lemma_1}.
	Through Lemma \ref{lemma_2}, we know that $\text{tr}({\bf{ A}}'{\bf{\tilde F}}')\le 0$.
	Therefore, we need the eigenvalue of ${{\bf{\tilde F}}[1:{N_{{\text{ue}}}}{N_{{\text{tx}}}},\;1:{N_{{\text{ue}}}}{N_{{\text{tx}}}}]}$
	to be larger to satisfy the constraint $C_{2,1}$ of $\mathcal{P}2$.
	However, this is contrary to the objective of $\mathcal{P}2$.
	Therefore, $i_{\text{max}}\le N_{\text{ue}}N_{\text{tx}}$ holds true and Theorem \ref{theorem_2} is proved.

	\section{Proof of Theorem 3}\label{appendix_3}
	Construct $\mathbf{a}_1$, which is a rearrangement of $\mathbf{a}$ in ascending order of its absolute values.
	Construct $\mathbf{a}_2$, which is arbitrary permutation of $\mathbf{a}$.
	Through \textit{rearrangement inequality}, we know 
	\begin{align}
		|\mathbf{a}_1^H\mathbf{a}_1|
		\ge
		|\mathbf{a}_1^H\mathbf{a}_2|.
	\end{align}
	Furthermore, we can obtain
	\begin{align}\label{eq_Appendix_3}
		\text{Re}(\mathbf{a}_1^H\mathbf{a}_1)
		\ge
		\text{Re}(\mathbf{a}_1^H\mathbf{a}_2),
	\end{align}
	since $|\mathbf{a}_1^H\mathbf{a}_1|=\text{Re}(\mathbf{a}_1^H\mathbf{a}_1)$ and $|\mathbf{a}_1^H\mathbf{a}_2|\ge \text{Re}(\mathbf{a}_1^H\mathbf{a}_2)$.	
	By comparing $\mathbf{a}_1^H\mathbf{a}_1$ in (\ref{eq_Appendix_3}) to $\text{tr}(\mathbf{A})$ in Theorem \ref{theorem_3}, and comparing $\mathbf{a}_1^H\mathbf{a}_2$ in (\ref{eq_Appendix_3}) to $\text{tr}(\mathbf{AP}_k)$ in Theorem \ref{theorem_3},
	Theorem \ref{theorem_3} can be proved.

\end{appendices}
\bibliographystyle{IEEEtran}
\bibliography{ref}

% \printbibliography{}

\vfill

\end{spacing}
\end{document}